\documentclass[11 pt]{article}
\date{}

\usepackage[section]{placeins}
\usepackage{graphicx,subfigure}
\usepackage{amsmath, amsthm, amssymb}
\usepackage{dsfont}
\usepackage{cite}
\usepackage{url}
\usepackage{color}

\newcommand{\dff}{\stackrel{\scriptscriptstyle\triangle}{=}}

\newcommand{\bbe}{\mathbb{E}}
\newcommand{\sinr}{{\sf SINR}}
\newcommand{\snr}{{\sf SNR}}
\newcommand{\dof}{{\sf dof}}

\newcommand{\bA}{\boldsymbol{A}}
\newcommand{\bH}{\boldsymbol{H}}

\newcommand{\balpha}{\boldsymbol{\alpha}}

\newcommand{\bI}{\boldsymbol{I}}

\newcommand{\bx}{\boldsymbol{x}}
\newcommand{\by}{\boldsymbol{y}}
\newcommand{\bz}{\boldsymbol{z}}
\newcommand{\bR}{\boldsymbol{R}}

\newcommand{\bX}{\boldsymbol{X}}

\newcommand{\N} {{\mathcal N}}

\newtheorem{definition}{Definition}
\newtheorem{theorem}{Theorem}
\newtheorem{lemma}{Lemma}
\newtheorem{coro}{Corollary}
\theoremstyle{definition}
\newtheorem{remark}{Remark}

\newcommand{\med}{\;|\;}

\newcommand\dotgeq{{\;\overset \cdot \geq\;}}
\newcommand\dotleq{{\overset \cdot \leq}}

\begin{document}

\title{$(n,K)$-user Interference Channels: Degrees of Freedom}

\author{Ali Tajer\and Xiaodong Wang\footnote{Electrical Engineering Department, Columbia University, New York, NY 10027.}}


\maketitle


\begin{abstract}
We analyze the gains of opportunistic communication in multiuser interference channels. Consider a fully connected $n$-user Gaussian interference channel. At each time instance only $K\leq n$ transmitters are allowed to be communicating with their respective receivers and the remaining $(n-K)$ transmitter-receiver pairs remain inactive. For finite $n$, if the transmitters can acquire channel state information (CSI) and if all channel gains are bounded away from zero and infinity, the seminal results on interference alignment establish that for any $K$ {\em arbitrary} active pairs the total number of spatial degrees of freedom per orthogonal time and frequency domain is $\frac{K}{2}$. Also it is noteworthy that without transmit-side CSI the interference channel becomes interference-limited and the degrees of freedom is 0. In {\em dense} networks ($n\rightarrow\infty$), however, as the size of the network increase, it becomes less likely to sustain the bounding conditions on the channel gains. By exploiting this fact,  we show that when $n$ obeys certain scaling laws, by {\em opportunistically} and {\em dynamically} selecting the $K$ active pairs at each time instance, the number of degrees of freedom can exceed $\frac{K}{2}$ and in fact can be made arbitrarily close to $K$. More specifically when all transmitters and receivers are equipped with one antenna, then the network size scaling as $n\in\omega(\snr^{d(K-1)})$ is a {\em sufficient} condition for achieving $d\in[0,K]$ degrees of freedom. Moreover, achieving these degrees of freedom does not necessitate the transmitters to acquire channel state information. Hence, invoking opportunistic communication in the context of interference channels leads to achieving higher degrees of freedom that are not achievable otherwise\footnote{$d\in(0,K]$ with no transmitter CSI and $d\in(\frac{K}{2},K]$ with transmitter CSI.}. We extend the results for multi-antenna Gaussian interference channels.
\end{abstract}


\newpage
\section{Introduction}
\label{sec:intro}

The emerging wireless networks are interference-limited due to the increasing demands for multimedia communications and ambitious spectral efficiency targets. The interference channel, a core component of such systems, becomes of paramount importance and has attracted significant recent interest. While the full extent of interference channels is still unknown, there exists a rich literature, spanning from the initial work by Shannon~\cite{Shannon:61} and the best achievable rate region~\cite{HK:IT81} to the most recent developments on the approximate capacity of two-user interference channels~\cite{Etkin:IT08} and the notion of interference alignment~\cite{Cadambe:IT08, Maddah:IT08} for the $K$-user interference channel.

Although the recent advances still dos not fully characterize the capacity region, they provide very insightful results exposing some fundamental limits of the interference channel. In particular, the results in~\cite{Cadambe:IT08, Maddah:IT08} indicate that in a fully connected $K$-user interference channel, through interference alignment each user can {\em almost surely} achieve as much as half of its interference-free capacity at the asymptote of large $\snr$. Besides the certain merits of analyzing the interference-channel as an stand-alone system, it is also imperative to obtain insight into their performance when they are embedded in a larger network. A good example of such larger networks are multi-cell downlink systems that can be considered as a generalization of the interference channel where each transmitter serves multiple receivers via spatial multiplexing.

In this paper we consider an interference channel embedded in a {\em dense} wireless network and analyze the degrees of freedom achievable for the interference channel of interest. In dense wireless networks the resources might be inadequate for serving all users concurrently. While being an impediment, such a situation nevertheless brings about the opportunity of tracking network state fluctuations and dynamically identifying and allocating the resources (power and bandwidth) to the {\em best} links at each time. Such notion of resource allocation, known as {\em opportunistic} communication, can effectively combat undesired channel variations as its performance relies on the peak, rather than average, channel conditions. Furthermore, the performance improves as the number of users increases, as it becomes more likely to encounter stronger links. Opportunistic communication has been investigated for multiple access channels and broadcast channels~\cite{Knopp:ICC95, Viswanath:IT02} and its gain, often referred to as multiuser diversity gain, is quantified as the {\em double-logarithmic} growth of the sum-rate with the size of the network~\cite{Sharif:COM07}.

In this paper we aim to investigate the gains of opportunistic communication in interference channels, and specifically its effect on the number of degrees of freedom. The interference alignment results indicated that under certain conditions on channel gains, for a fully connected single-antenna $K$-user interference channel the maximum number of degrees of freedom is $\frac{K}{2}$~\cite{Cadambe:IT08, Maddah:IT08}. We show that by leveraging opportunistic communication, under certain conditions on the network size, it is possible to recover the lost half of bandwidth for each user and achieves $K$ degrees of freedom. Quantifying the number of degrees freedom in a large network essentially entails assessing the sum-rate of the $K$-user interference channels operating at both asymptotes of large $\snr$ and network size $n$. We obtain some {\em sufficient} conditions on how the network size should scale in order to ensure capturing any degrees of freedom of interest $d$. More specifically, when the network size scales as

\begin{equation*}
    n\in\omega\left(\snr^{d(K-1)}\right)
\end{equation*}
it is sufficient to guarantee achieving $d$ degrees of freedom via opportunistically activating the best set of $K$ users at-a-time. In Section~\ref{sec:prel} after describing the system model we provide some detailed discussions on where the gains offered by opportunistic communication in the interference channels are originated from. We summarize the results for single-antenna and multi-antenna interference channels in Section~\ref{sec:main}. The proofs of the main results are provided in Sections~\ref{sec:SD:siso} and \ref{sec:SD:mimo} for the single-antenna and multiple-antenna cases, respectively, and Section~\ref{sec:conclusion} concludes the paper.

\section{Motivation and Objective}
\label{sec:prel}

\subsection{Channel Model}
\label{sec:model}

For given positive integers $n,K$, where $n\gg K$, we define an $(n,K)$-user interference channel as follows. Consider a wireless network consisting of $n$ pairs of transmitters and receivers, where each transmitter intends to communicate exclusively with its designated receiver. During each time slot, only $K$ transmitter-receiver pairs are allowed to be communicating, constituting a $K$-user Gaussian interference channel, while the remaining $(n-K)$ pairs remain inactive. We also use the following conventions: $A(n)\dff\{1,\dots,n\}$ denotes the set of the indices of all transmitter-receiver pairs, and $V_t\subset A(n)$ contains the indices of the transmitter-receiver pairs that are active during time slot $t$, where by definition $\forall t\in\mathbb{R}_+:\;| V_t|=K$. We denote the $K$ elements of the set $V_t$ (the indices of the active users) by $\{v_1,\dots,v_K\}$.

We assume that the transmitters and receivers are equipped with $N$ antennas. The wireless channel from transmitter $v\in A(n)$ to receiver $u\in A(n)$ during time slot $t$ undergoes Rayleigh fading and is denoted by $\bH_{u,v}[t]\in\mathbb{C}^{N\times N}$. The elements of the channel matrix $\bH_{u,v}[t]$ are independent and identically distributed (i.i.d.) with distribution $\N_{\mathbb{C}}(0,1)$\footnote{$\N_{\mathbb{C}}(a,b)$ denotes a circularly symmetric complex Gaussian distribution with mean $a$ and variance $b$.}. During time slot $t$, the received signal by the receiver $u\in V_t$ is given by
\begin{equation}\label{eq:model}
    \forall u\in V_t:\qquad \by_u[t]=\sum_{v\in V_t}\sqrt{\gamma_{u,v}}\;\bH_{u,v}[t]\;\bx_{v}[t]+\bz_u[t]\ ,
\end{equation}
where $\gamma_{u,v}\in\mathbb{R}_+$ accounts for the path-loss along channel $\bH_{u,v}[t]$ for all $u,v\in A(n)$.
Also $\bx_v[t]\in\mathbb{C}^{N\times 1}$ is the signal vector transmitted by transmitter $v$ and $\bz_u[t]$ denotes the additive white Gaussian noise term with i.i.d. entries distributed as $\N_{\mathbb{C}}(0,1)$. The active users during time slot $t$ have the average transmission power $\snr$, i.e., $\forall v\in V_t: \bbe\{\bx^H_v[t]\bx_v[t]\}= \snr$. Finally, we define $R_u$ as the rate that the $u^{th}$ transmitter-receiver pair can sustain reliably and define the rate vector for the set of active users $V_t$ as
\begin{equation*}
    \bR_{ V_t}\dff[R_{v_1},\dots,R_{v_K}]^T\ .
\end{equation*}
For each $V_t\subset A(n)$, we denote the union of all achievable rate regions by $C_{ V_t}$ (the capacity region).
\begin{remark}
\label{rem:dep}
It should be noted that the rate vector $\bR_{ V_t}$ and the capacity region $C_{ V_t}$ are functions of $\snr$ and channel realizations $\bH_{u,v}$ for $u,v\in V_t$, where their explicit dependence on $\snr$ and $\bH_{u,v}$ is omitted for the convenience in notations.
\end{remark}

\subsection{Motivation}
\label{sec:motivation}
The results on interference alignment~\cite{Cadambe:IT08} establish that in a fully-connected single-antenna $K$-user interference channel, the pre-log factor of the sum-capacity at the asymptote of large $\snr$ (degrees of freedom) is $\frac{K}{2}$. This result relies on the assumption that the channel gains are bounded away from zero and infinity. More specifically, if for any arbitrary set of user pairs $V_t\subset A(n)$ that constitute a fully-connected $K$-user interference channel we have
\begin{equation}\label{eq:H:cond}
   \exists H_{\min},\;H_{\max}\in\mathbb{R}^+\quad\mbox{such that}\quad\forall u,v\in V_t:\quad 0<H_{\min}\leq |\bH_{u,v}[t]|\leq H_{\max}<+\infty\ ,
\end{equation}
then the sum-rate capacity $C_{V_t,\;\rm sum}$ in the high $\snr$ regime has the pre-log factor $\frac{K}{2}$. In other words,
\begin{equation}\label{eq:C:cond}
   \mbox{if}\quad\forall u,v\in V_t,\;\; \bH_{u,v}[t]\;\;\mbox{satisfies \eqref{eq:H:cond}}\quad\Rightarrow\quad C_{V_t,\;\rm sum}=\frac{K}{2}\log\snr+Y_{V_t}\ ,
\end{equation}
where $Y_{V_t}$ is a function of $\snr$ and $\{\bH_{u,v}\}_{u,v\in V_t}$, for which we have
\begin{equation*}
    \lim_{\snr\rightarrow\infty}\frac{Y_{V_t}}{\log\snr}=0\ .
\end{equation*}
This result essentially implies that although $C_{ V_t}$ is a function of the channels $\{\bH_{u,v}\}_{u,v\in V_t}$ (Remark~\ref{rem:dep}), as long as the condition in \eqref{eq:H:cond} is satisfied, the pre-log factor (number of degrees of freedom) is independent of $\{\bH_{u,v}\}_{u,v\in V_t}$. One immediate conclusion is that under condition \eqref{eq:H:cond}, for all $n\choose K$ possible choices for $V_t$, and irrespective of any strategy for selecting $V_t$, the degrees of freedom always is $\frac{K}{2}$. Note that for small or moderate network size $n$, and for any given set of user pairs $V_t$, by selecting $H_{\min}$ and $H_{\max}$ arbitrarily small and large, respectively, we can ensure that all channel realizations satisfy the bounding conditions in \eqref{eq:H:cond} {\em almost surely}.

In dense networks, on the other hand, as the network size grows ($n\rightarrow\infty$) the likelihood that some channel channels violate the bounding constraints \eqref{eq:H:cond} increases. As it will be made clear later in the paper, under {\em certain conditions} on the size of $n$, there will be instances that the channel realizations for some groups of users $V_t\subset A(n)$ violate the bounding constraints \eqref{eq:H:cond}. As it will be shown, in such instances for the sum-rate capacity, as opposed to \eqref{eq:C:cond}, we have
\begin{equation*}
    \mbox{if}\quad \exists u,v\in V_t,\quad\mbox{such that}\quad \bH_{u,v}[t]\;\;\mbox{does {\em not} satisfy \eqref{eq:H:cond}}\quad\Rightarrow\quad C_{V_t,\;\rm sum}=X_{V_t}\log\snr+\tilde Y_{V_t}\ ,
\end{equation*}
where $X_{V_t}$ and $\tilde Y_{V_t}$ are functions of $\snr$ and $\{\bH_{u,v}\}_{u,v\in V_t}$. For $\tilde Y_{V_t}$ we have
\begin{equation*}
    \lim_{\snr\rightarrow\infty}\frac{\tilde Y_{V_t}}{\log\snr}=0\ ,
\end{equation*}
and $X_{V_t}$, depending on the structure of the channels, can lie anywhere within the interval $[0,K]$. For instance $X_{V_t}=K$ degrees of freedom is achievable in the very unlikely, but not impossible, extreme situation where all direct channels (connecting each transmitter to its designated receiver) are very strong and the cross (interfering) links are extremely weak. In such an extreme situation the system is essentially equivalent to $K$ (almost) non-interfering parallel channels that give rise to $K$ degrees of freedom.

As the network size $n$ increases, the network becomes richer in the sense that it offers more diverse channel realizations. Consequently, the likelihood that we encounter a set of users $V_t$ for which the degrees of freedom $X_{V_t}$ exceeds $\frac{K}{2}$, and possibly approaches $K$, increases. Motivated by this premise, we aim to characterize how $n$ should scale in order to guarantee attaining any arbitrary degree of freedom in the interval $(\frac{K}{2},K]$. We offer a few definitions as follows. For any channel realization $\{\bH_{u,v}\}_{u,v}$ and for any given set of users $V_t\subset A(n)$, we define the degrees of freedom achievable when $V_t$ is the set of active users as
\begin{eqnarray}\label{eq:dof:V}
   \dof_{V_t}(n,K) \dff  X_{V_t} = \limsup_{\snr\rightarrow\infty}\left[\frac{1}{\log\snr} \sup_{\bR_{ V_t}\in C_{ V_t}}\boldsymbol{1}^T\cdot\bR_{ V_t}\right]\ ,
\end{eqnarray}
where $\boldsymbol{1}_{K\times 1}$ is the vector of all ones. By opportunistically opting for the set of users $V_t$ that yield the largest $\dof_{V_t}(n,K)$ over all possible choices of $V_t$, for the $(n,K)$-user interference channel we also define
\begin{equation}\label{eq:dof*}
    \dof^*(n,K)\dff \max_{ V_t\subset A(n):\;| V_t|=K}\;\dof_{V_t}(n,K)\ .
\end{equation}
Since $\dof_{V_t}(n,K)=X_{V_t}$, and consequently $\dof^*(n,K)$, are functions of $\{\bH_{u,v}\}_{u,v\in V(t)}$, they are random variables inhering their randomness from the randomness of the channel coefficients. Therefore, we define the {\em ergodic} degrees of freedom for the $(n,K)$-user interference channel as the mean of $\dof^*(n,K)$ over the ensemble  of {\em all} possible channel realizations. This ergodic degrees of freedom, denoted by $\dof(n,K)$, is given by
\begin{eqnarray}\label{eq:dof}
    \dof(n,K)& \dff & \bbe_{\bH}\left[\dof^*(n,K)\right]=\bbe_{\bH}\left[\max_{ V_t\subset A(n):\;| V_t|=K}\; \dof_{V_t}(n,K)\right]\ .
\end{eqnarray}
Characterizing $\dof(n,K)$ essentially requires tracking channel state fluctuations over time and dynamically activating the $K$ {\em best} transmitter-receiver pairs that yield the {\em largest} number of degrees of freedom at each time instance. Our objective is to characterize the achievable degrees of freedom of the $(n,K)$-user interference channel in the asymptote of large network sizes, i.e.,
\begin{equation*}
    \lim_{n\rightarrow\infty}\dof(n,K)\ .
\end{equation*}

\subsection{Objective}
\label{sec:SD}

Motivated by the premise that increasing the network size in conjunction with opportunistic selection of the active users enables achieving higher degrees of freedom, we aim to characterize the scaling law for the network size in order to guarantee achieving $d$ degrees of freedom for any arbitrary $d\in[0,K]$. We assume that {\em all} receivers employ single-user decoders, where each receiver recovers its designated signal via linear filtering and treating the rest of interfering signals as Gaussian noise. Single-user decoders, being suboptimal receivers, provide lower bounds on the optimal degrees of freedom achievable for the $(n,K)$-user interference channels. Given that the users employ single-user decoders, we derive the requirements for the network size $n$ that suffice to ensure capturing any degrees of freedom of interest. Invoking the suboptimality of single-user decoders, these requirements in turn provide some {\em sufficient} condition on the scaling laws of the network size for achieving any arbitrary degrees of freedom in the interval $[0,K]$.

Let us denote the rates achievable via single-user decoding for the set of active users $V_t$ by $\bR^{\rm sd}_{V_t}$. Similar to \eqref{eq:dof:V}, for any channel realization $\{\bH_{u,v}\}_{u,v}$ and for any given set of users $V_t\subset A(n)$, the number of degrees of freedom upon employing single-user decoders is denoted by
\begin{eqnarray}\label{eq:dof:V:sd}
   \dof^{\rm sd}_{V_t}(n,K) \dff \limsup_{\snr\rightarrow\infty}\left[\frac{1}{\log\snr}\; \boldsymbol{1}^T\cdot\bR^{\rm sd}_{ V_t}\right]\ .
\end{eqnarray}
Also, similar to \eqref{eq:dof*} and \eqref{eq:dof} we define the instantaneously maximum and the  ergodic degrees of freedom for the $(n,K)$-user interference channel with single-user decoding as
\begin{eqnarray}\label{eq:dof:sd*}
    \dof^*_{\rm sd}(n,K)& \dff & \max_{ V_t\subset A(n):\;| V_t|=K}\; \dof^{\rm sd}_{V_t}(n,K)\ ,
\end{eqnarray}
and
\begin{eqnarray}\label{eq:dof:sd}
    \dof_{\rm sd}(n,K)& \dff & \bbe[\dof^*_{\rm sd}(n,K)]=\bbe_{\bH}\left[\max_{ V_t\subset A(n):\;| V_t|=K}\; \dof^{\rm sd}_{V_t}(n,K)\right]\ .
\end{eqnarray}
As $\bR^{\rm sd}_{V_t}\in C_{V_t}$ we immediately have
\begin{equation}\label{eq:rate_sd}
    \boldsymbol{1}^T\cdot\bR^{\rm sd}_{ V_t}\leq \sup_{\bR_{ V_t}\in C_{ V_t}}\boldsymbol{1}^T\cdot\bR_{ V_t}\quad\Rightarrow\quad  \lim_{n\rightarrow\infty}\dof_{\rm sd}(n,K)\leq \lim_{n\rightarrow\infty}\dof(n,K)\ .
\end{equation}

\section{Main Results}
\label{sec:main}

\subsection{Single-antenna Users}
\label{sec:SD:siso:main}
We provide the main results of the paper in this section and relegate the proofs and the ensuing discussions to Sections~\ref{sec:SD:siso}~and~\ref{sec:SD:mimo}. We start by considering the case where all transmitters and receivers are equipped with one antenna, i.e., $N=1$. Then from \eqref{eq:model} the signal-to-interference-plus-noise ratio ($\sinr$) at the receivers of the active pairs is given by
\begin{equation}\label{eq:sinr}
    \forall u\in V_t:\quad \sinr_u=\frac{\gamma_{u,u}|\bH_{u,u}|^2}{\sum_{v\in V_t,\;v\neq u}\gamma_{u,v}|\bH_{u,v}|^2+\snr^{-1}}\ .
\end{equation}
Since the receivers employ single-user decoders, the rates sustained by the users in $ V_t$ are given by
\begin{equation}\label{eq:Rsd}
    \forall u\in V_t:\quad R_{u}^{\rm sd}\dff\log(1+\sinr_{u})\ .
\end{equation}
By recalling \eqref{eq:dof:V:sd} we have
\begin{eqnarray*}
     \dof^{\rm sd}_{V_t}(n,K) & = & \limsup_{\snr\rightarrow\infty}\left[\frac{1}{\log\snr}\sum_{u\in V_t}\log(1+\sinr_{u})\right]\ .
\end{eqnarray*}
Note that for any given set of active users $V_t$, the rates of the active users clearly depends {\em only} on the channels between the active users. Consequently, $\dof^{\rm sd}_{V_t}$ is a random variable\footnote{For convenience we sometimes abbreviate $\dof^{\rm sd}_{V_t}(n,K)$ as $\dof^{\rm sd}_{V_t}$.} inhering its randomness from {\em only} the channels of the users that their indices are included in $V_t$. Therefore, corresponding to the $n\choose K$ possible choices for $V_t$, we have a sequence of random variables $\{\dof^{\rm sd}_{V_t}\}_{V_t}$ of length $n\choose K$. According to the definition of $\dof_{\rm sd}(n,K)$ given in~\eqref{eq:dof:sd}, characterizing $\dof_{\rm sd}(n,K)$ requires knowing the distribution of the largest order statistic of the sequence $\{\dof^{\rm sd}_{V_t}\}_{V_t}$, i.e., $\max_{V_t}\dof^{\rm sd}_{V_t}$. Note that due to the statistical independence of the channel coefficients, for any two arbitrary sets $V_t$ and $\tilde V_t$ we have
\begin{equation}\label{eq:disjoint}
    \mbox{if}\;\;V_t\cap\tilde V_t=\emptyset\quad\Rightarrow\quad \dof^{\rm sd}_{V_t}\quad\mbox{and}\quad\dof^{\rm sd}_{\tilde V_t}\quad\mbox{are statistically {\em independent}}\ .
\end{equation}
Moreover, when $V_t$ and $\tilde V_t$ are {\em not} disjoint, their common users induce some correlation, i.e.,
\begin{equation}\label{eq:common}
    \mbox{if}\;\;V_t\cap\tilde V_t\neq\emptyset\quad\Rightarrow\quad \dof^{\rm sd}_{V_t}\quad\mbox{and}\quad\dof^{\rm sd}_{\tilde V_t}\quad\mbox{are statistically {\em correlated}}\ .
\end{equation}
Hence, $\{\dof^{\rm sd}_{V_t}\}_{V_t}$ is a sequence of {\em correlated} random variables. Moreover, due to the different path-losses that different users experience, the elements of $\{\dof^{\rm sd}_{V_t}\}_{V_t}$ are non-identically distributed. Therefore, characterizing $\dof_{\rm sd}(n,K)$ requires obtaining the largest order statistics of a sequence of non-identically distributed and correlated random variables, which seems intractable (especially since there is no specific correlation structure). Nevertheless, we find some lower and upper bounds on the distribution of the largest order statistics of $\{\dof^{\rm sd}_{V_t}\}_{V_t}$, which in turn offer lower and upper bounds on the achievable number of degrees of freedom $\dof_{\rm sd}(n,K)$.

For the upper bound, as we will discuss in detail in Section~\ref{sec:SD:siso}, we use the properties of {\em exchangeable} sequence of random variables and use the result of {\em de Finetti}'s theorem \cite{Heath:76} in order to find a bound on the distribution of the largest order statistic of a {\em correlated} sequence of random variables. For obtaining the lower bound on $\dof_{\rm sd}(n,K)$ we partition the set of $n$ transmitter-receiver pairs to $M\dff\lfloor\frac{n}{K}\rfloor$ {\em disjoint} sets $U_1,\dots,U_M$ each consisting of $K$ transmitter-receiver pairs. Optimizing $\dof^{\rm sd}{V_t}(n,K)$ over such partitions instead of all possible partitions clearly incurs a loss in the achievable degrees of freedom and hence provides a lower bound on it. In other words,
\begin{eqnarray*}
    \max_{ V_t\subset A(n):\;| V_t|=K}\; \dof^{\rm sd}_{V_t}(n,K)\geq \max_{m\in\{1,\dots,M\}}\dof^{\rm sd}_{U_m}(n,K)\ ,
\end{eqnarray*}
which provides that
\begin{equation}\label{eq:dof:sd_partition}
   \dof_{\rm sd}(n,K) = \bbe_{\bH}\left[\max_{ V_t\subset A(n):\;| V_t|=K}\; \dof^{\rm sd}_{V_t}(n,K)\right]\geq \bbe_{\bH}\left[\max_{m\in\{1,\dots,M\}}\dof^{\rm sd}_{U_m}(n,K)\right] \ .
\end{equation}
Similar to what mentioned earlier in \eqref{eq:disjoint}, since the sets $U_1,\dots,U_M$ are disjoint, the random variables $\{\dof^{\rm sd}_{U_1},\dots,\dof^{\rm sd}_{U_M}\}$ become independent. Such independence enables obtaining the distribution of the largest order statistic of the sequence of random variables $\{\dof^{\rm sd}_{U_1},\dots, \dof^{\rm sd}_{U_M}\}$. The main result for the single-antenna $(n,K)$-user interference channel is offered in the following theorem.
\begin{theorem}
\label{th:sd}
For the $(n,K)$-user interference channel with single-antenna users and single-user decoders at the receivers we have
\begin{equation}\label{eq:th:sd}
    \min\left(K,\frac{\xi_n}{K-1}\right)\leq \lim_{n\rightarrow\infty}\dof_{\rm sd}(n,K)=\lim_{n\rightarrow\infty}\bbe_{\bH}[\dof^*_{\rm sd}(n,K)]\leq K\cdot\min\left(1,2\xi_n\right)\ ,
\end{equation}
where $\xi_n$ is defined as
\begin{equation}\label{eq:zeta}
    \xi_n\dff\lim_{\snr\rightarrow\infty}\frac{\log n}{\log\snr}\ .
\end{equation}
Also, {\em almost surely} we have
\begin{equation}\label{eq:th:sd2}
    \min\left(K,\frac{\xi_n}{K-1}\right)\leq \lim_{n\rightarrow\infty}\dof^*_{\rm sd}(n,K)\leq K\cdot\min\left(1,2\xi_n\right)\ ,
\end{equation}

\end{theorem}
The theorem above establishes lower and upper bounds on $\lim_{n\rightarrow\infty}\dof_{\rm sd}(n,K)$. By noting that the single-user decoders are {\em sub-optimal} receivers, we immediately find that the lower bound in \eqref{eq:th:sd} is also a lower bound on $\lim_{n\rightarrow\infty}\dof(n,K)$, i.e., the degrees of freedom of the $(n,K)$-user interference channel in the asymptote of large $n$. Hence, by leveraging this lower bound we can obtain a {\em sufficient} condition on the scaling law of the network size for achieving any arbitrary degrees of freedom in the interval $[0,K]$.

Note that that that achieving the degrees of freedom characterized by the theorem above do not necessitate any transmit-side channel state information (CSI). The CSI is necessary for only calculating the sum-rate achievable for all possible sets of active users $V_t$. Therefore, it suffices that such CSI is only revealed to the receivers. Moreover, for achieving the lower bound in Theorem~\ref{th:sd} the receivers are required to obtain only some {\em local} CSI. More specifically, based on the construction of the proofs for the lower bounds, we group the $n$ pairs of transmitters-receivers into subgroups each containing $K$ users and select the best subgroup as the active set of users. For this purpose each subgroup of users have to obtain only {\em local} CSI in order to identify the sum-rate achievable for them. Eventually the subgroup that the largest achievable sum-rate is selected as the set of active users.

\begin{coro}
\label{cor:sufficient}
For the $(n,K)$-user interference channel with single-antenna users and single-user decoders at the receivers, a sufficient condition for achieving $d\in[0,K]$ degrees of freedom is
\begin{equation*}
    \xi_n\geq d(K-1)\ .
\end{equation*}
\end{coro}

It is noteworthy that for finite network size $n$, without transmit-side CSI the interference channel is interference-limited and the degrees of freedom is 0, whereas for large networks, depending on the network size, it can be up to $K$. On the other hand, when the transmitters can acquire CSI, interference alignment always offers $\frac{K}{2}$ degrees of freedom {\em almost surely}. Therefore, with the transmit-side CSI, the region of more significance is $d\in(\frac{K}{2},K]$ that is not achievable without invoking opportunistic selection of the active users.

In the next corollary, we also provide a necessary condition on the scaling law of the network size for achieving $d$ degrees of freedom. This necessary condition, however, unlike the sufficient condition in  Corollary~\ref{cor:sufficient} is restricted to  single-user decoders and it is expected that for more advanced receivers, the necessary conditions on the scaling of $n$ is stringent.
\begin{coro}
\label{cor:necessary}
For the $(n,K)$-user interference channel with single-antenna users and single-user decoders at the receivers, a necessary condition for achieving $d\in[0,K]$ degrees of freedom is
\begin{equation*}
    \xi_n\geq \frac{d}{2K}\ .
\end{equation*}
\end{coro}
Note that when the network size $n$ is fixed, i.e., when $\xi_n=0$, by employing single-user decoders (and no interference alignment) the network becomes interference-limited. In other words, the $\sinr$s and the rates will be saturating by increasing $\snr$ and consequently we expect to have $d=0$ degrees of freedom.

\subsection{Multi-antenna Users}
\label{sec:SD:mimo:main}

Next we generalize the results to the case that the transmitters and receivers are equipped with $N\in\mathbb{N}$ antennas. Each user can achieve a degrees of freedom up to $N$ and the degrees of freedom for the $(n,K)$-user interference channel can be any point within the interval $[0,NK]$. Similar to the single-antenna case, the objective is to characterize the scaling laws that warrant capturing any arbitrary degrees of freedom in the interval $[0,NK]$. For any arbitrary set of active users $V_t\subset A(n)$, and for all active users $u\in V_t$, let us define the $N\times (K-1)N$ matrix $\bH_{u,V_t}$ by concatenating the channel matrices of {\em all} users interfering with user $u$, i.e.,
\begin{equation}\label{eq:H}
     \forall u\in V_t:\qquad \bH_{u,V_t}\dff\left[{\sqrt{\gamma_{u,v_1}}}\;\bH_{u,v_1},\dots, {\sqrt{\gamma_{u,v_K}}}\;\bH_{u,v_K}\right]\  .
\end{equation}
Based on the signal model \eqref{eq:model}, upon employing single-user decoding, the rate of the active pairs at time instance $t$ is
\begin{equation}\label{eq:R}
     \forall u\in V_t:\qquad R^{\rm sd}_u=\log\det\left[ \bI+\snr\cdot\gamma_{u,u}\;\bH_{u,u}^H \left(\bI+\snr\cdot\bH_{u,V_t}\bH_{u,V_t}^H\right)^{-1}\bH_{u,u}\right]\ .
\end{equation}
Similar to the single-antenna case the random variables $\dof^{\rm sd}_{V_t}$ and $\dof^{\rm sd}_{\tilde V_t}$ are {\em independent} when the sets $V_t$ and $\tilde V_t$ are disjoint, and are correlated otherwise. For the same intractability reasons, we resort to obtaining lower bounds on the degrees of freedom. For this purpose, we derive two different lower bounds on the degrees of freedom and take their union to obtain a unified lower bound. As the first lower bound, we directly apply the result of Theorem~\ref{th:sd} by pairing-up transmit and receive antennas of each user and treating each pair as one independent transmitter-receiver pair. More specifically, we consider the $N$ antennas of each transmitter as one independent transmitter and pair it with one of the receive antennas of the designated receiver. In this way we essentially transform the $N$-antenna $(n,K)$-user interference channel into a single-antenna $(nN,KN)$-user interference channel. According to Theorem~\ref{th:sd} we can find a lower bound on the degrees of freedom.

As the second lower bound, we again consider the same partitioning technique through which we can characterize a sequence of {\em independent} random variables with tractable distribution for the largest order statistics. The main result for the multi-antenna $(n,K)$-user interference channel is presented in the following theorem.
\begin{theorem}
\label{th:sd:mimo}
For the $(n,K)$-user interference channel with $N$ antenna at each transmitter and receiver, and single-user decoders at the receivers we have
\begin{equation}\label{eq:th:sd:mimo}
    \min\left(NK,\max\left(\frac{\xi_n}{NK-1},\zeta_n\right)\right)\leq \lim_{n\rightarrow\infty}\dof_{\rm sd}(n,K)=\lim_{n\rightarrow\infty}\bbe_{\bH}[\dof^*(n,K)]\ ,
\end{equation}
where
\begin{equation*}
    \xi_n\dff\lim_{\snr\rightarrow\infty}\frac{\log n}{\log\snr}\quad\mbox{and}\quad \zeta_n\dff\frac{\sqrt{(K-2)^2N^2+4\xi_n}-(K-2)N}{2}\ .
\end{equation*}
Also, {\em almost surely} we have
\begin{equation}\label{eq:th:sd:mimo2}
    \min\left(NK,\max\left(\frac{\xi_n}{NK-1},\zeta_n\right)\right)\leq \lim_{n\rightarrow\infty}\dof^*_{\rm sd}(n,K)\ .
\end{equation}

\end{theorem}
Similar to the single-antenna setup, we can find a {\em sufficient} condition on the scaling law of the network size $n$, in order to guarantee achieving any arbitrary degrees of freedom $d$.
\begin{coro}
\label{cor:necessary}
For the $(n,K)$-user interference channel with  $N$ antenna at each transmitter and receiver, and single-user decoders at the receivers, a sufficient condition for achieving $d\in[0,K]$ degrees of freedom is
\begin{equation*}
    \xi_n\geq \left\{
    \begin{array}{ll}
      d^2+d(K-2)N\ , & \mbox{if}\quad 2N-1\geq d\ , \\
      d(NK-1)\ , & \mbox{if}\quad 2N-1\leq d\ .
    \end{array}\right.
\end{equation*}
\end{coro}

\section{Preliminaries}
\label{sec:remark}

In this section we briefly provide some definitions and propositions that are instrumental and frequently referred to throughout the rest of the paper.
\begin{definition}
\label{rem:exp}
We say two functions $f(\snr)$ and $g(\snr)$ are {\em exponentially equal} when
\begin{equation*}
  \lim_{\snr\rightarrow\infty}\frac{\log f(\snr)}{\log g(\snr)}=1\ .
\end{equation*}
We use the convention $f(\snr)\doteq g(\snr)$ to denote such exponential equality and define the operators $\dotgeq$ and $\dotleq$ accordingly. We also state that the {\em exponential order} of $f(\snr)$ is $d$ when $f(\snr)\doteq \snr^d$.
\end{definition}
\begin{definition}
\label{rem:alpha}
For the random variable $X$ distributed as $\N_\mathbb{C}(0,1)$ define
\begin{equation}\label{eq:alpha}
  \alpha_X\dff-\frac{\log|X|^2}{\log\snr}\ .
\end{equation}
Clearly, $|X|^2$ is exponentially distributed with mean $1$. The cumulative distribution function (cdf) of $\alpha_X$ is given by
\begin{equation*}
  F_{\alpha_X}(\alpha)=P(\alpha_X\leq \alpha)=P(|X|^2\geq \snr^{-\alpha})=\exp\left(-\snr^{-\alpha}\right)\ ,
\end{equation*}
and the probability density function (pdf) of $\alpha_X$ is thereof given by
\begin{equation*}
  f_{\alpha_X}(\alpha)=\log(\snr)\;\snr^{-\alpha}\; \exp\left(-\snr^{-\alpha}\right)\ .
\end{equation*}
It can be readily verified that we also have the following exponential equality for the pdf of $\alpha_X$ \cite{Azarian:IT05}
\begin{equation}\label{eq:alpha:pdf:exp}
  f_{\alpha_X}(\alpha)\doteq\snr^{-\alpha}\cdot\mathds{1}_{\{\alpha\geq 0\}}\ ,
\end{equation}
where $\mathds{1}_A:\mathbb{R}\rightarrow\{0,1\}$ is the {\em indicator} function defined as
\begin{equation*}
  \mathds{1}_A\dff\left\{
  \begin{array}{cc}
    1, & A \mbox{ is true},\\
    0, & A \mbox{ is false}.
  \end{array}\right.
\end{equation*}

\end{definition}
\begin{definition}
\label{rem:beta}
For the random variable $X$ distributed as $\N_\mathbb{C}(0,1)$ define
\begin{equation*}
  \beta_X\dff\min(\alpha_X,1)\ ,
\end{equation*}
where $\alpha_X$ is defined in \eqref{eq:alpha}. The cdf of $\beta_X$ is
\begin{equation}\label{eq:beta:cdf}
  F_{\beta_X}(\beta)=\exp\left(-\snr^{-\beta}\right)\cdot\mathds{1}_{\{\beta< 1\}}+\mathds{1}_{\{\beta\geq 1\}}\ ,
\end{equation}
and its pdf is exponentially equal to
\begin{equation*}
  f_{\beta_X}(\beta)\doteq\snr^{-\beta}\cdot\mathds{1}_{\{0\leq\beta< 1\}} +(1-\snr^{-1})\delta(\beta-1)\doteq\snr^{-\beta}\cdot\mathds{1}_{\{0\leq\beta< 1\}}\ ,
\end{equation*}
where $\delta(\cdot)$ denotes Dirac's delta function.
\end{definition}
\begin{remark}
\label{rem:d}
For $d_1,\dots,d_m\in\mathbb{R}$ we have
\end{remark}
\begin{equation}\label{eq:d}
    \sum_{i=1}^m\snr^{d_i}\doteq\max_i\snr^{d_i}\doteq \snr^{\max_id_i}\ .
\end{equation}

\begin{remark}
\label{rem:region}
If the probability density functions of the {\em independent} random variables $X_1,\dots,X_m$, are {\em exponentially} equal to
\begin{equation*}
    \forall i\in\{1,\dots,m\}:\quad f_{X_i}\doteq \snr^{-x_i}\cdot\mathds{1}_{\{0\leq x_i\leq t\}}\ ,
\end{equation*}
then the probability that
$\bX\dff(X_1,\dots,X_m)$ belongs to the region $B$ is exponentially equal to
\begin{equation*}
    P(\bX\in B)\doteq\snr^{-b}\ ,
\end{equation*}
where
\begin{equation}\label{eq:region}
    b=\inf_{\bX\in \tilde B}\sum_{i=1}^mx_i\ ,\quad\mbox{and}\quad \tilde B=\{\bX\med \bX\in B\quad\mbox{and}\quad \boldsymbol{0} \preceq\bX\preceq t\cdot\boldsymbol{1}\}\ .
\end{equation}
\end{remark}
\begin{remark}
\label{rem:a}
For the positive real value $a\in\mathbb{R}_+$ we have
\begin{equation*}
   1-\exp\left(1-\snr^{-a}\right)\doteq \snr^{-a}\ .
\end{equation*}
\end{remark}
\begin{remark}
\label{rem:ab}
For positive real values $a, b\in\mathbb{R}_+$ and for the functions $f,\;g:\;\mathbb{R}_+\rightarrow\mathbb{R}_+$, if $f(\snr)\doteq\snr^{-a}$ and $g(\snr)\doteq\snr^{b}$, then
\begin{equation}\label{eq:ab}
    \lim_{\snr\rightarrow\infty}\left(1-f(\snr)\right)^{g(\snr)}=\left\{
    \begin{array}{cc}
      0, & \mbox{if}\;\;b>a ,\\
      1, & \mbox{if}\;\;b<a .
    \end{array}
    \right.\ .
\end{equation}
\end{remark}

\section{Single-antenna Users}
\label{sec:SD:siso}
The proof consists of three main steps. In the first step, for each arbitrary set of active users $V_t$ we formulate the achievable degrees of freedom $\dof^{\rm sd}_{V_t}(n,K)$ as a function of the exponential orders (Remark~\ref{rem:exp}) of the channel coefficients of the users with their indices included in $V_t$. In the second step, by using the results of Definitions~\ref{rem:alpha}~and~\ref{rem:beta} we obtain the probability distribution of $\dof^{\rm sd}_{V_t}(n,K)$ for each arbitrary $V_t$. In the third step, finally, by using the distribution of $\dof^{\rm sd}_{V_t}(n,K)$ we offer lower and upper bounds on the distributions of the largest order statistics of the sequence $\{\dof^{\rm sd}_{V_t}\}_{V_t}$, which consequently provide lower and upper bounds on $\dof_{\rm sd}(n,K)$.

\subsection{Characterizing $\dof^{\rm sd}_{V_t}(n,K)$}
From \eqref{eq:dof:V:sd} and \eqref{eq:Rsd} recall that
\begin{eqnarray}\label{eq:dof:sd:Vt:proof}
     \dof^{\rm sd}_{V_t}(n,K) & = & \limsup_{\snr\rightarrow\infty}\left[\frac{1}{\log\snr}\sum_{u\in V_t}\log(1+\sinr_{u})\right]\ .
\end{eqnarray}
For the set of active users $V_t$ let us define
\begin{equation}\label{eq:alpha:uv}
    \forall u,v\in V_t:\quad \alpha_{u,v}\dff -\frac{|\bH_{u,v}|^2}{\log\snr}\quad\Rightarrow\quad \gamma_{u,v}|\bH_{u,v}|^2=\gamma_{u,v}\;\snr^{-\alpha_{u,v}}\doteq\snr^{-\alpha_{u,v}}\ .
\end{equation}
Note that due to the statistical independence of $\{\bH_{u,v}\}_{u,v}$, their associated exponential orders $\{\alpha_{u,v}\}_{u,v}$ also become independent. By recalling $\sinr_u$, as given in \eqref{eq:sinr}, and by invoking the exponential equalities in \eqref{eq:alpha:uv} we obtain the following exponential equality.
\begin{eqnarray}
  \nonumber \forall u\in V_t:\quad 1+\sinr_u & \doteq & \snr^0+\frac{\snr^{-\alpha_{u,u}}}{\sum_{v\in V_t,\;v\neq u}\snr^{-\alpha_{u,v}}+\snr^{-1}}\\
  \nonumber & \overset{\eqref{eq:d}}{\doteq} & \snr^0+\frac{\snr^{-\alpha_{u,u}}}{\max\left\{\max_{v\in V_t,\;v\neq u}\left\{\snr^{-\alpha_{u,v}}\right\},\; \snr^{-1}\right\}} \\
  \label{eq:sinr:exp1} & \doteq & \snr^0+\frac{\snr^{-\alpha_{u,u}}}{\snr^{-\beta_u({V_t})}}\ ,
\end{eqnarray}
where we have defined
\begin{equation}\label{eq:beta:Vt}
    \forall u\in V_t:\quad \beta_u({V_t})\dff \min\Big\{\min_{v\in V_t,\;v\neq u}\left\{\alpha_{u,v}\right\},\;1\Big\}=\min_{v\in V_t,\;v\neq u}\Big\{\min\left\{\alpha_{u,v},\;1\right\}\Big\}\ .
\end{equation}
It is noteworthy that for any set of active users $V_t$ and any transmitter-receive pair $u$, the random variable $\beta_u(V_t)$ is shaped up by the channel coefficients of all channels from transmitters $v\neq u$, where $v\in V_t$, to receiver $u$. Therefore, it can be readily verified that for $(V_t,u)\neq (\tilde V_t,\tilde u)$, the random variables $\beta_u(V_t)$ and $\beta_{\tilde u}(\tilde V_t)$ are statistically independent. Next, equations \eqref{eq:sinr:exp1} and \eqref{eq:beta:Vt} give rise to
\begin{eqnarray}\label{eq:sinr:exp2}
  \nonumber \forall u\in V_t:\quad 1+\sinr_u & \doteq & \snr^0+ \snr^{\beta_u({V_t})-\alpha_{u,u}}\\
  & \doteq & \snr^{(\beta_u({V_t})-\alpha_{u,u})^+}\ ,
\end{eqnarray}
where we have defined $(x)^+=\max(0,x)$. The definition of the exponential equality (Definition~\ref{rem:exp}) in conjunction with \eqref{eq:sinr:exp2} provide that
\begin{equation}\label{eq:sinr:exp3}
    \forall u\in V_t:\quad \lim_{\snr\rightarrow\infty}\frac{R^{\rm sd}_u}{\log\snr}\overset{\eqref{eq:Rsd}}{=}\lim_{\snr\rightarrow\infty}\frac{\log(1+\sinr_u)}{\log\snr}=(\beta_u({V_t})-\alpha_{u,u})^+\ ,
\end{equation}
where $(\beta_u({V_t})-\alpha_{u,u})^+$ is a random variable inheriting its randomness from the the channel coefficients $\{\bH_{u,v}\}_{v\in V_t}$ through their associated exponential orders $\{\alpha_{u,v}\}_{v\in V_t}$. Equations ~\eqref{eq:dof:sd:Vt:proof} and \eqref{eq:sinr:exp3} yield that the number of degrees of freedom for the set of active users $V_t$ when they deploy single-user decoding is given by
\begin{equation}\label{eq:dof:V:sd1}
    \dof^{\rm sd}_{V_t}(n,K) = \limsup_{\snr\rightarrow\infty}\left[\frac{1}{\log\snr}\; \boldsymbol{1}^T\cdot\bR^{\rm sd}_{ V_t}\right]= \sum_{u\in V_t}(\beta_u({V_t})-\alpha_{u,u})^+\ .
\end{equation}

\subsection{Distribution of $\dof^{\rm sd}_{V_t}(n,K)$}
Next we aim to obtain the distribution of $\dof^{\rm sd}_{V_t}(n,K)$, as characterized in \eqref{eq:dof:V:sd1}, through finding the distributions of its summands $(\beta_u({V_t})-\alpha_{u,u})^+$. We define a new random variable corresponding to each summand of \eqref{eq:dof:V:sd1}.
\begin{equation}\label{eq:Z}
    \forall V_t\subset A(n),\;\;\forall u\in V_t:\quad Z_u(V_t)\dff (\beta_u({V_t})-\alpha_{u,u})^+\ .
\end{equation}
The following lemma provides the exponential order of the probability density function of $Z_u(V_t)$.
\begin{lemma}\label{lmm:Z}
For the probability density function (pdf) of $Z_u(V_t)$, denoted by $f_Z(z)$, we have
\begin{equation}\label{eq:Z:cdf}
    f_Z(z)\doteq  \snr^{-(K-1)z}\cdot\mathds{1}_{\{0\leq z\leq 1\}}\ .
\end{equation}
\end{lemma}
\begin{proof}
See Appendix \ref{app:lmm:Z}.
\end{proof}
Note that while for distinct choices of $(V_t,u)\neq(\tilde V_t,\tilde u)$ the random variables $Z_u(V_t)$ and $Z_{\tilde u}(\tilde V_t)$ are {\em not} identically distributed (due to different path losses of the channels), their probability density functions exhibit identical exponential orders. For notational convenience we define
\begin{equation}\label{eq:X}
    \forall V_t\subset A(n):\quad X_{V_t}\dff\dof^{\rm sd}_{V_t}(n,K)=\sum_{u\in V_t}Z_u(V_t)\ ,
\end{equation}
where $Z_u(V_t)$ is defined in \eqref{eq:Z}. Next, by using the exponential equality on the pdf of $Z_u(V_t)$ provided in Lemma~\ref{lmm:Z}, we proceed to find the distribution of $X_{V_t}$ in the next lemma.
\begin{lemma}\label{lmm:X}
For the cumulative density function (cdf) of $X_{V_t}$, denoted by $F_X(x)$, we have
\begin{equation}\label{eq:X:cdf}
    1-F_X(x) \doteq \snr^{-(K-1)x}\cdot \mathds{1}_{\{0\leq x\leq K\}}\ .
\end{equation}
\end{lemma}
\begin{proof}
As mentioned earlier, for distinct choices $u\neq v$, the random variables $\beta_u(V_t)$ and $\beta_{v}(V_t)$ are statistically independent. By further taking into account the statistical independence among the elements of $\{\alpha_{u,v}\}_{u,v\in V_t}$  it can be readily verified that
\begin{equation*}
    \forall u\neq v\in V_t:\quad Z_u(V_t)\;\;\mbox{and}\;\;Z_v(V_t)\;\;\mbox{are also statistically independent}\ .
\end{equation*}
Therefore, the joint pdf of the $K$ random variables $\{Z_u(V_t)\}_{u\in V_t}$ is simply the products of their marginal pdfs, i.e.,
\begin{equation*}
    \prod_{u\in V_t}f_Z(z_u)\ .
\end{equation*}
Consequently,
\begin{eqnarray}\label{eq:X:cdf3}
   1-F_X(x)= P(X_{V_t}> x)=P\left(\sum_{u\in V_t}Z_u(V_t)> x\right)=\int_B \prod_{u\in V_t}(f_Z(z_u)\;dz_u)\ ,
\end{eqnarray}
where
\begin{equation*}
    B=\Big\{\{z_u\}_u\;\Big|\; z_u=Z_u(V_t)\quad\mbox{and}\quad \sum_{u\in V_t}Z_u(V_t)\geq x\Big\}\ .
\end{equation*}
By invoking Lemma \ref{lmm:Z} and substituting $f_Z(z_u)$ in \eqref{eq:X:cdf3} as
\begin{equation*}
    f_Z(z_u)\doteq \snr^{-(K-1)z_u}\cdot\mathds{1}_{\{0\leq z_u\leq 1\}}
\end{equation*}
we find that for $0\leq x\leq K$ we have
\begin{eqnarray}\label{eq:X:cdf4}
  \nonumber 1-F_X(x)&\doteq& \int_B\snr^{-\sum_{u\in V_t}(K-1)z_u}\prod_{u\in V_t}\mathds{1}_{\{0\leq z_u\leq 1\}}\;dz_u\\
  &=&\int_{\tilde B} \snr^{-\sum_{u\in V_t}(K-1)z_u}\prod_{u\in V_t}\;dz_u\ ,
\end{eqnarray}
where
\begin{equation*}
    \tilde B=\Big\{\{z_u\}_u\;\Big|\;z_u=Z_u(V_t)\quad\mbox{and}\quad \sum_{u\in V_t}Z_u(V_t)\geq x\quad\mbox{and}\quad  \forall u\in V_t:\; 0\leq z_u\leq 1\Big\}\ .
\end{equation*}
Hence, from \eqref{eq:X:cdf4} and by taking into account Remark~\ref{rem:region}, for $0\leq x\leq K$ we obtain
\begin{equation}\label{eq:X:cdf5}
    1-F_X(x)\doteq \snr^{-b}\qquad\mbox{where}\qquad b=\inf_{\{z_u\}\in \tilde B}\sum_{u\in V_t}z_{u}= (K-1)x\ .
\end{equation}
Finally note that as discussed in the proof of Lemma~\ref{lmm:Z}, the random variable $Z_u(V_t)$ lies in the interval $[0,1]$, and consequently, the range of $X_{V_t}=\sum_{u\in V_t}Z_u(V_t)$ is $[0,K]$. Therefore, from \eqref{eq:X:cdf5} we get
\begin{equation*}
    1-F_X(x)\left\{
    \begin{array}{ll}
      =1, & x< 0, \\
      \doteq \snr^{-(K-1)x}, & 0\leq x\leq K, \\
      =0, & x>K,
    \end{array}
    \right.\ ,
\end{equation*}
which is the desired result.
\end{proof}

\subsection{Bounds on $\dof_{\rm sd}(n,K)$}
To this end we have obtained the distribution of $\dof^{\rm sd}_{V_t}(n,K)$ and by recalling \eqref{eq:dof:sd*} and \eqref{eq:dof:sd} we have
\begin{equation}\label{eq:dof:partition}
   \dof_{\rm sd}(n,K) = \bbe_{\bH}[\dof^*_{\rm sd}(n,K)]=\bbe_{\bH}\left[\max_{ V_t\subset A(n):\;| V_t|=K}\; \dof^{\rm sd}_{V_t}(n,K)\right] \ ,
\end{equation}
which requires finding the distribution of the largest order statistics of the sequence $\{\dof^{\rm sd}_{V_t}\}_{V_t}$ which, as discussed in Section~\ref{sec:main}, consists of {\em correlated} random variables. We next obtain tractable bounds on the desired distribution, which in turn provide lower and upper bounds on $\dof_{\rm sd}(n,K)$.

\begin{description}
  \item[1) Lower Bound on $\dof_{\rm sd}(n,K)$:]\textcolor{white}{a}\\
      For obtaining the lower bound on $\dof_{\rm sd}(n,K)$ we partition the set of $n$ transmitter-receiver pairs to $M\dff\lfloor\frac{n}{K}\rfloor$ {\em disjoint} sets $U_1,\dots,U_M$ each consisting of $K$ transmitter-receiver pairs as follows,
      \begin{equation}\label{eq:U}
        \forall i\in\{1,\dots,\lfloor n/K\rfloor\}:\quad U_i\dff \{K(i-1)+1,\dots,Ki\}\ .
      \end{equation}
      By noting that $|U_i|=K$,  we clearly have
      \begin{equation}\label{eq:dof:sd:l1}
        \dof^*(n,K)=\max_{ V_t\subset A(n):\;| V_t|=K}X_{V_t}\geq \max_{i\in\{1,\dots,M\}}X_{U_i} \quad\overset{\eqref{eq:dof:partition}}{\Rightarrow}\quad \dof_{\rm sd}(n,K)\geq \bbe_{\bH}\left[\max_{i\in\{1,\dots,M\}}X_{U_i} \right]\ .
      \end{equation}
      Similar to what mentioned earlier in \eqref{eq:disjoint}, since the sets $U_1,\dots,U_M$ are disjoint, the random variables $X_{U_i}$ and $X_{U_j}$, inheriting their randomness from the channels of the users in $U_i$ and $U_j$, respectively, are statistically independent for $i\neq j$. By enforcing such independence, for the cdf of $\max_iX_{U_i}$, denoted by $F_X^{\max}(x)$, we have
      \begin{eqnarray}\label{eq:dof:sd:l2}
        F_X^{\max}(x) = P(\max_i X_{U_i}\leq x) = \big(F_X(x)\big)^{\lfloor n/K\rfloor}= \Big(1-\big(1-F_X(x)\big)\Big)^{\lfloor n/K\rfloor}\ .
      \end{eqnarray}
      Next, in order to characterize $F_X^{\max}(x)$ we use the results of Remark \ref{rem:ab}. For this purpose for any $x\in[0,K]$ we define the functions $f_x(\snr)$ and $g_x(\snr)$ as follows.
      \begin{equation*}
        f_x(\snr)\dff 1-F_X(x)\quad \overset{\eqref{eq:X:cdf}}{\Rightarrow}\quad f_x(\snr)\doteq \snr^{-(K-1)x}\quad\overset{{\eqref{eq:ab}}}{{\Rightarrow}}\quad a=(K-1)x\ ,
      \end{equation*}
      and
      \begin{equation*}
         g_x(\snr)\dff{\lfloor n/K\rfloor}\quad\Rightarrow\quad \lim_{\snr\rightarrow\infty}\frac{\log g_x(\snr)}{\log\snr}=\lim_{\snr\rightarrow\infty}\frac{\log n}{\log\snr}\overset{\eqref{eq:zeta}}{=}\zeta_n \quad \overset{\eqref{eq:ab}}{\Rightarrow}\quad b=\zeta_n\ .
      \end{equation*}
      Therefore, from Remark \ref{rem:ab} and \eqref{eq:dof:sd:l2} we find that for $x\in[0,K]$
      \begin{equation*}
        \lim_{\snr\rightarrow\infty}F_X^{\rm max}(x)= \lim_{\snr\rightarrow\infty}\left(1-f_x(\snr)\right)^{g_x(\snr)}=\left\{
        \begin{array}{cc}
          0, & \mbox{if}\;\;\zeta_n>(K-1)x, \\
          1, & \mbox{if}\;\;\zeta_n<(K-1)x,
        \end{array}
        \right.
      \end{equation*}
      or equivalently for any $x\in[0,K]$
      \begin{equation}\label{eq:dof:sd:l4}
        \lim_{\snr\rightarrow\infty}F_X^{\rm max}(x)= \left\{
        \begin{array}{cc}
          0, & \mbox{if}\;\;x<\frac{\zeta_n}{K-1}, \\
          1, & \mbox{if}\;\;x>\frac{\zeta_n}{K-1}.
        \end{array}
        \right.
      \end{equation}
      Moreover, by noting that $X_{V_t}\in[0,K]$ we consequently have $\max_iX_{U_i}\in[0,K]$, which immediately provides
      \begin{equation}\label{eq:dof:sd:l3}
        \forall x\geq K:\quad F^{\rm max}_X(x)=1\ .
      \end{equation}
     Equations \eqref{eq:dof:sd:l4} and \eqref{eq:dof:sd:l3} together give rise to
     \begin{equation}\label{eq:dof:sd:l5}
        \lim_{\snr\rightarrow\infty}F_X^{\rm max}(x)= \left\{
        \begin{array}{ll}
          0, & \mbox{if}\;\;x<\min\left(K,\frac{\zeta_n}{K-1}\right), \\
          1, & \mbox{if}\;\;x>\min\left(K,\frac{\zeta_n}{K-1}\right).
        \end{array}
        \right.
      \end{equation}
      Some simple manipulations yield that for the pdf of $\max_iX_{U_i}$ we have
      \begin{equation}\label{eq:dof:sd:l6}
        \lim_{\snr\rightarrow\infty} f_X^{\rm max}(x)=\delta\left(x-\min\left(K,\frac{\zeta_n}{K-1}\right)\right)\ .
      \end{equation}
      Finally, from \eqref{eq:dof:sd:l1} and \eqref{eq:dof:sd:l6} we find that
      \begin{eqnarray}\label{eq:dof:sd:l7}
        \nonumber \lim_{n\rightarrow\infty}\dof_{\rm sd}(n,K)&\overset{\eqref{eq:zeta}}{=}&\lim_{\snr\rightarrow\infty}\dof_{\rm sd}(n,K)\\
        \nonumber&\overset{\eqref{eq:dof:sd:l1}}{\geq}& \lim_{\snr\rightarrow\infty}\bbe_{\bH}\left[\max_{i\in\{1,\dots,M\}}X_{U_i} \right]\\
        \nonumber&=& \lim_{\snr\rightarrow\infty} \int_0^Kxf_X^{\rm max}(x)\;dx\\
        \nonumber&\overset{\rm (g)}{=}& \int_0^Kx\lim_{\snr\rightarrow\infty} f_X^{\rm max}(x)\;dx\\
        \nonumber&=& \int_0^K x\cdot \delta\left(x-\min\left(K,\frac{\zeta_n}{K-1}\right)\right)\;dx\\
        &=& \min\left(K,\frac{\zeta_n}{K-1}\right)\ .
      \end{eqnarray}
      Exchanging the limit and integral in (g) is justified according to Lebesgue's dominated convergence Theorem~ \cite{bartle:B1}.
  \item[2) Upper Bound on $\dof_{\rm sd}(n,K)$:]\textcolor{white}{a}\\
      We start by providing the following lemma for the {\em exchangeable} sequences of random variables. A finite or infinite sequence of random variables $\{X_1,\dots,X_n\}$ is called {\em exchangeable} if for any possible {\em finite} permutation of the indices $1,\dots,n$ (any permutation that keeps all but a finite number of indices fixed) the joint pdf of the permutated sequence is equal to that of the original sequence.
      \begin{lemma}\label{lmm:exchange}
      For an exchangeable sequence of random variables with {\em identical} and {\em not necessarily independent} distributions we have
      \begin{equation}\label{eq:exchange}
        P\left(\max_i X_i\leq x\right)\geq \left[P(X_i\leq x)\right]^n\ .
      \end{equation}
      \end{lemma}
      \begin{proof}
      See Appendix \ref{app:lmm:exchange}.
      \end{proof}
      From the definition of $X_{V_t}=\dof^{\rm sd}_{V_t}(n,K)$ given \eqref{eq:X} we can find the following upper bound on $X_{V_t}$
      \begin{equation*}
        \forall V_t\subset A(n):\quad X_{V_t}=\sum_{u\in V_t}Z_u(V_t)\leq K\cdot \max_{u\in V_t} Z_u(V_t)\ .
      \end{equation*}
      Based on the definition of $\dof_{\rm sd}(n,K)$ given in \eqref{eq:dof:sd} we find that
      \begin{eqnarray}\label{eq:dof:sd1:max}
        \nonumber \dof_{\rm sd}(n,K) & \leq & K\cdot\bbe_{\bH}\left[\max_{ V_t\subset A(n):\;| V_t|=K}\; \max_{u\in V_t} Z_u(V_t)\right]\\
        &=& K\cdot\bbe_{\bH}\bigg[\max_{u\in A(n)}\underset{\dff W_u}{\underbrace{\max_{ V_t\subset A(n):\;| V_t|=K,\; u\in V_t}\;  Z_u(V_t)}}\bigg] \ .
      \end{eqnarray}
      Due to the symmetry involved, the sequence of random variables $\{W_u\}_{u\in A(n)}$ are exchangeable. Therefore, by invoking Lemma \ref{lmm:exchange} we find that
      \begin{eqnarray}\label{eq:dof:sd1:max2}
        \nonumber P\left(\max_{u\in A(n)}W_u\leq w\right)&\overset{\eqref{eq:exchange}}{\geq}&  [P(W_u\leq w)]^n\\
        \nonumber &=&\left[P\left(\max_{ V_t\subset A(n):\;| V_t|=K,\; u\in V_t}\;  Z_u(V_t)\leq w\right)\right]^n\\
        &\overset{\eqref{eq:Z}}{=}& \left[P\left(\max_{ V_t\subset A(n):\;| V_t|=K,\; u\in V_t}\;  (\beta_u(V_t)-\alpha_{u,u})^+\leq w\right)\right]^n\ .
      \end{eqnarray}
      By further defining
      \begin{equation}\label{eq:beta:uv2}
          \forall u,v\in V_t:\quad \beta_{u,v}\dff \min\{\alpha_{u,v},\;1\}\ ,
      \end{equation}
      and recalling that we had defined \eqref{eq:beta:Vt}
      \begin{equation*}
          \forall u\in V_t:\quad \beta_u({V_t})\dff \min\Big\{\min_{v\in V_t,\;v\neq u}\left\{\alpha_{u,v}\right\},\;1\Big\}=\min_{v\in V_t,\;v\neq u}\Big\{\min\left\{\alpha_{u,v},\;1\right\}\Big\}\ ,
      \end{equation*}
      we get the following connection between between $\beta_u(V_t)$ and $\beta_{u,v}$
      \begin{equation}\label{eq:beta:Vt22}
           \forall V_t\subset A(n),\;\;\forall u\in V_t:\quad\beta_u({V_t})=\min_{v\in V_t,\;v\neq u}\beta_{u,v}\ .
      \end{equation}
      By substituting $\beta_u(V_t)$ with its equivalent given above, from \eqref{eq:dof:sd1:max2} we have
      \begin{eqnarray}\label{eq:dof:sd1:max2_1}
        \nonumber P\left(\max_{u\in A(n)}W_u\leq w\right)&{\geq}& \left[P\left(\max_{ V_t\subset A(n):\;| V_t|=K,\; u\in V_t}\;  \left(\min_{v\in V_t,\;v\neq u}\beta_{u,v}-\alpha_{u,u}\right)^+\leq w\right)\right]^n \\
        \nonumber &=&  \left[P\left(\max_{ V_t\subset A(n):\;| V_t|=K,\; u\in V_t}\;  \min_{v\in V_t,\;v\neq u} \left(\beta_{u,v}-\alpha_{u,u}\right)^+\leq w\right)\right]^n\\
        \nonumber & \geq &  \left[P\left(\max_{ V_t\subset A(n):\;| V_t|=K,\; u\in V_t}\;  \max_{v\in V_t,\;v\neq u} \left(\beta_{u,v}-\alpha_{u,u}\right)^+\leq w\right)\right]^n\\
        \nonumber & = &  \left[P\left(\max_{v\neq u} \left(\beta_{u,v}-\alpha_{u,u}\right)^+\leq w\right)\right]^n\\
        & = &  \left[P\left( \left(\beta_{u,v}-\alpha_{u,u}\right)^+\leq w\right)\right]^{n(n-1)}\quad\mbox{for}\;\;v\neq u\ .
      \end{eqnarray}
      where the last step holds due to the statistical independence between any two elements of $\{(\beta_{u,v}-\alpha_{u,u})^+\}_{u,v}$. Note that for any two random variables $X$ and $Y$, if the cdf of $X$ uniformly dominates $Y$, i.e., $F_X(x)\geq F_Y(y)$, or equivalently $\int x\;d(F_X(x))\leq \int y\;d(F_Y(y))$. By applying this observation from \eqref{eq:dof:sd1:max2_1} we obtain
      \begin{eqnarray}\label{eq:dof:sd1:max3}
        \nonumber \dof_{\rm sd}(n,K)&\overset{\eqref{eq:dof:sd1:max}}{\leq }& K\cdot \bbe_{\bH}\left[\max_{u\in A(n)} W_u\right]= K\int_w\;w\cdot\;d\left[P\left(\max_{u\in A(n)}W_u\leq w\right)\right]\\
        & \overset{\eqref{eq:dof:sd1:max2}}{\leq }& K\int_w\;w\cdot d\left[\left[P\left( \left(\beta_{u,v}-\alpha_{u,u}\right)^+\leq w\right)\right]^{n(n-1)}\right]\quad\mbox{for}\;\;v\neq u\ .
      \end{eqnarray}

      In the next step we find the distribution of $(\beta_{u,v}-\alpha_{u,u})^+$ as formalized in the following lemma.
      \begin{lemma}\label{lmm:beta_alpha}
      For the CDF of $(\beta_{u,v}-\alpha_{u,u})^+$ we have the following exponential equality
      \begin{equation}\label{eq:beta_alpha}
        1-P\left( \left(\beta_{u,v}-\alpha_{u,u}\right)^+\leq w\right)\doteq \snr^0\cdot \mathds{1}_{\{w< 0\}}+\snr^{-w}\cdot \mathds{1}_{\{0\leq w \leq 1\}}\ .
      \end{equation}
      \end{lemma}
      \begin{proof}
      See Appendix \ref{app:lmm:beta_alpha}.
      \end{proof}
      Given the Lemma above, for $w\in[0,1]$ we use the result of Remark \ref{rem:ab} by setting
      \begin{equation*}
        f_w(\snr)\dff 1-P\left( \left(\beta_{u,v}-\alpha_{u,u}\right)^+\leq w\right)\quad \overset{\eqref{eq:beta_alpha}}{\Rightarrow}\quad f_w(\snr)\doteq \snr^{-w}\quad\overset{{\eqref{eq:ab}}}{{\Rightarrow}}\quad a=w\ ,
      \end{equation*}
      and
      \begin{equation*}
         g_w(\snr)\dff n(n-1)\quad\Rightarrow\quad \lim_{\snr\rightarrow\infty}\frac{\log g_w(\snr)}{\log\snr}\overset{\eqref{eq:zeta}}{=}2\zeta_n \quad \overset{\eqref{eq:ab}}{\Rightarrow}\quad b=2\zeta_n\ .
      \end{equation*}
      Therefore, from Remark \ref{rem:ab} we find that for  $w\in[0,1]$
      \begin{align}\label{eq:dof:sd1:max6}
        \nonumber \lim_{\snr\rightarrow\infty}&\left[P\left( \left(\beta_{u,v}-\alpha_{u,u}\right)^+\leq w\right)\right]^{n(n-1)}\\
        &= \lim_{\snr\rightarrow\infty}\left(1-f_w(\snr)\right)^{g_w(\snr)}=\left\{
        \begin{array}{cc}
          0, & \mbox{if}\;\;2\zeta_n>w, \\
          1, & \mbox{if}\;\;2\zeta_n<w.
        \end{array}
        \right.\
      \end{align}
      Taking into account the range of $\left(\beta_{u,v}-\alpha_{u,u}\right)^+$ (proof of Lemma~\ref{lmm:beta_alpha}, Equation \eqref{eq:dof:sd1:max4}) in conjunction with  \eqref{eq:dof:sd1:max6} establish that
     \begin{equation}\label{eq:dof:sd1:max7_1}
        \lim_{\snr\rightarrow\infty}\left[P\left( \left(\beta_{u,v}-\alpha_{u,u}\right)^+\leq w\right)\right]^{n(n-1)}= \left\{
        \begin{array}{ll}
          0, & \mbox{if}\;\;w<\min(2\zeta_n,1), \\
          1, & \mbox{if}\;\;w>\min(2\zeta_n,1),
        \end{array}
        \right.\
      \end{equation}
      which in turn provides that the pdf of $\max_iX_{U_i}$ in the high $\snr$ regime is
      \begin{equation}\label{eq:dof:sd1:max7}
         \lim_{\snr\rightarrow\infty}\;d\left[P\left( \left(\beta_{u,v}-\alpha_{u,u}\right)^+\leq w\right)\right]^{n(n-1)}=\delta\left(w-\min(2\zeta_n,1)\right)\ .
      \end{equation}
      Therefore, by following the same line of argument as in the case for the lower bound, from \eqref{eq:dof:sd1:max3} and \eqref{eq:dof:sd1:max7} we find that
      \begin{eqnarray*}
        \lim_{n\rightarrow\infty}\dof_{\rm sd}(n,K)&\overset{\eqref{eq:zeta}}{=}&\lim_{\snr\rightarrow\infty}\dof_{\rm sd}(n,K)\\
        &\overset{\eqref{eq:dof:sd1:max3}}{\leq}& \lim_{\snr\rightarrow\infty} K\int_w\;w\cdot d\left[\left[P\left( \left(\beta_{u,v}-\alpha_{u,u}\right)^+\leq w\right)\right]^{n(n-1)}\right]\quad\mbox{for}\;\;v\neq u\ \\
        &=& K\int_w\;w\cdot\lim_{\snr\rightarrow\infty} d\left[\left[P\left( \left(\beta_{u,v}-\alpha_{u,u}\right)^+\leq w\right)\right]^{n(n-1)}\right]\quad\mbox{for}\;\;v\neq u\ \\
        &\overset{\eqref{eq:dof:sd1:max7}}{=}& K \int_0^1w\cdot\delta\left(w-\min(2\zeta_n,1)\right)\;dw\\
        &=& K\cdot\min(2\zeta_n,1)\ ,
      \end{eqnarray*}
      which is the desired upper bound.
\end{description}

\subsection{Proof of Equation \eqref{eq:th:sd2}}

Recall from the definition of $\{U_1,\dots,U_M\}$ given in \eqref{eq:U} that
\begin{equation}\label{eq:dof:sd:l1_copy}
  \dof^*(n,K)=\max_{ V_t\subset A(n):\;| V_t|=K}X_{V_t}\geq \max_{i\in\{1,\dots,M\}}X_{U_i} \ .
\end{equation}
Moreover, for the pdf of $\max_iU_i$ denoted by $f^{\max}_X(x)$ we have
\begin{equation*}
  \lim_{\snr\rightarrow\infty} f_X^{\rm max}(x)=\delta\left(x-\min\left(K,\frac{\zeta_n}{K-1}\right)\right)\ ,
\end{equation*}
which consequently indicates that
\begin{equation}\label{eq:dof:sd:l6_copy}
    P\left[\lim_{\snr\rightarrow\infty}\max_iU_i=\min\left(K,\frac{\zeta_n}{K-1}\right)\right]=1\ .
\end{equation}
Equations \eqref{eq:dof:sd:l1_copy} and \eqref{eq:dof:sd:l6_copy} yield that
\begin{equation}\label{eq:dof:sd:as}
    P\left[\lim_{\snr\rightarrow\infty}\dof^*(n,K)\geq \min\left(K,\frac{\zeta_n}{K-1}\right)\right]=1\ ,
\end{equation}
which is the desired lower bound on $\dof^*(n,K)$ given in \eqref{eq:th:sd2}. Obtaining the upper bound follows the same line of argument.
\section{Multi-antenna Users}
\label{sec:SD:mimo}

Similar to the single-antenna case, the proof consists of three main steps. In the first step we try to characterize the number of degrees of freedom for any arbitrary given set $V_t$, i.e., $\dof^{\rm sd}_{V_t}(n,K)$ and we find a {\em lower} bound on it. Next we find the distribution of $\dof^{\rm sd}_{V_t}(n,K)$ and finally we find two bounds on the distribution of the largest order statistics of $\{\dof^{\rm sd}_{V_t}\}_{V_t}$, which collectively constitute a lower bound on $\dof_{\rm sd}(n,K)$.

\subsection{Characterizing $\dof^{\rm sd}_{V_t}(n,K)$}

Due to having distinct path-loss terms $\{\gamma_{u,v}\}$, the elements of $\bH_{u,V_t}$ while are statistically independently, do {\em not} have identical distributions. For tractability purposes we define another channel matrix, corresponding to which we find a tractable lower bound on $R^{\rm sd}_u$. Let us define
\begin{equation*}
    \forall u\in V_t:\qquad \gamma_{u,V_t}\dff\max_{v\in V_t,\;v\neq u}\gamma_{u,v}\ ,
\end{equation*}
and
\begin{equation}\label{eq:H_tilde}
     \forall u\in V_t:\qquad \tilde\bH_{u,V_t}\dff\left[{\sqrt{\gamma_{u,V_t}}}\;\bH_{u,v}\right]_{v\in V_t,\;v\neq u}\ .
\end{equation}
Since the receivers  employ single-user decoders, increasing the terms $\{\gamma_{u,v}\}$ for the interferers is equivalent to imposing more interference power on each active user, which in turn results in a reduction in the rates that the active can sustain reliably.
By invoking \eqref{eq:R} we obtain that
\begin{eqnarray}
     \nonumber \forall u\in V_t:\qquad R_u^{\rm sd} & \geq & \log\det\left[ \bI+\snr\cdot\gamma_{u,u}\;\bH_{u,u}^H \left(\bI+\snr\cdot\tilde\bH_{u,V_t}\tilde\bH_{u,V_t}^H\right)^{-1}\bH_{u,u}\right]\\
     \nonumber &=& \log\det\left[ \bI+ \left(\bI+\snr\cdot\tilde\bH_{u,V_t}\tilde\bH_{u,V_t}^H\right)^{-1} \snr\cdot\gamma_{u,u}\;\bH_{u,u}\bH_{u,u}^H\right]\\
     \nonumber &=& \log\;\frac{\det\left[\left(\bI+\snr\cdot\tilde\bH_{u,V_t}\tilde\bH_{u,V_t}^H\right) + \snr\cdot\gamma_{u,u}\;\bH_{u,u}\bH_{u,u}^H\right]} {\det\left[\bI+\snr\cdot\tilde\bH_{u,V_t}\tilde\bH_{u,V_t}^H\right]}\\
     \label{eq:R_lb} &\geq & \log\;\frac{\det\left[\bI+\snr\cdot\tilde\bH_{u,V_t}\tilde\bH_{u,V_t}^H\right]+\det \Big[\snr\cdot\gamma_{u,u}\;\bH_{u,u}\bH_{u,u}^H\Big]} {\det\left[\bI+\snr\cdot\tilde\bH_{u,V_t}\tilde\bH_{u,V_t}^H\right]}\\
     \label{eq:R_lb_2} &= & \log\left(1+\frac{\det\Big[\snr\cdot\gamma_{u,u}\;\bH_{u,u}\bH_{u,u}^H\Big]} {\det\left[\bI+\snr\cdot\tilde\bH_{u,V_t}\tilde\bH_{u,V_t}^H\right]}\right)\ ,
\end{eqnarray}
where the inequality in \eqref{eq:R_lb} holds by recalling that for the positive semi-definite matrices $A$ and $B$ we have $\det(A+B)\geq \det A+\det B$, and noting that $\tilde\bH_{u,V_t}\tilde\bH_{u,V_t}^H$ and $\bH_{u,u}\tilde\bH_{u,u}^H$ are positive semi-definite matrices.

Next, suppose $\mu_{u,u}^1\leq\mu_{u,u}^2\leq\dots\leq\mu_{u,u}^{N}$ are the ordered non-zero eigenvalues of $\bH_{u,u}\tilde\bH_{u,u}^H$. Similarly denote the ordered non-zero eigenvalues of $\tilde\bH_{u,V_t}\tilde\bH_{u,V_t}^H$ by $\lambda_{u,V_t}^1\leq\lambda_{u,V_t}^2\leq\dots\leq\lambda_{u,V_t}^{N}$. Define the exponential orders of these eigenvalues as follows,
\begin{equation}\label{eq:alpha:mu}
    \forall u\in V_t,\;\;\;\forall m\in\{1,\dots,N\}:\quad \alpha_{u,u}^m\dff -\frac{\log\mu_{u,u}^m}{\log\snr}\quad\mbox{and}\quad \beta_{u,V_t}^m\dff -\frac{\log\lambda^m_{u,V_t}}{\log\snr}\ .
\end{equation}
Therefore, from \eqref{eq:R_lb_2} we find that
\begin{eqnarray}\label{eq:R_lb_3}
    \nonumber \forall u\in V_t:\qquad R_u^{\rm sd} &\geq & \log\left(1+\frac{\prod_{m=1}^{N}\gamma_{u,u}\;\snr\cdot \mu_{u,u}^m} {\prod_{m=1}^{N}\left(1+\snr\cdot\lambda_{u,V_t}^m\right)}\right)\\
    \nonumber &\doteq& \log\left(1+\frac{\prod_{m=1}^{N}\snr^{1-\alpha_{u,u}^m}} {\prod_{m=1}^{N}\left(1+\snr^{1-\beta_{u,V_t}^m}\right)}\right)\\
    \nonumber &\doteq& \log\left(\snr^0+\frac{\prod_{m=1}^{N}\snr^{1-\alpha_{u,u}^m}} {\prod_{m=1}^{N}\left(\snr^0+\snr^{1-\beta_{u,V_t}^m}\right)}\right)\\
    \nonumber &\doteq& \log\left(\snr^0+\prod_{m=1}^{N}\snr^{1-\alpha_{u,u}^m}\prod_{m=1}^{N}\snr^{-(1-\beta_{u,V_t}^m)^+} \right)\\
    &\doteq& \log\left(\snr^0+\snr^{\sum_{m=1}^{N}(1-\alpha_{u,u}^m)}\cdot\snr^{-\sum_{m=1}^{N}(1-\beta_{u,V_t}^m)^+} \right)\ .
\end{eqnarray}
Therefore, by recalling the definition of the exponential equality (Definition~\ref{rem:exp}) we have
\begin{equation}\label{eq:R_lb_4}
    \forall u\in V_t:\qquad \lim_{\snr\rightarrow\infty}\frac{R_u^{\rm sd}}{\log\snr}\geq\left[\sum_{m=1}^{N}(1-\alpha_{u,u}^m)-\sum_{m=1}^{N}(1-\beta_{u,V_t}^m)^+\right]^+\ .
\end{equation}
The term above is a random variable that depends on the channel coefficients $\{\bH_{u,v}\}$ through the negative of their exponential orders. Considering \eqref{eq:dof:V:sd} and \eqref{eq:R_lb_4}, the number of degrees of freedom for the set of multiple-antenna active users $V_t$ when they deploy single-user decoding is given by
\begin{equation}\label{eq:dof:Vmimo:sd1}
    \dof^{\rm sd}_{V_t}(n,K) = \limsup_{\snr\rightarrow\infty}\left[\frac{1}{\log\snr}\; \boldsymbol{1}^T\cdot\bR^{\rm sd}_{ V_t}\right]\geq \sum_{u\in V_t}\left[\sum_{m=1}^{N}(1-\alpha_{u,u}^m)-\sum_{m=1}^{N}(1-\beta_{u,V_t}^m)^+\right]^+\ .
\end{equation}

\subsection{Distribution of $\dof^{\rm sd}_{V_t}(n,K)$}

Similar to \eqref{eq:Z} and \eqref{eq:X}, we define the following new random variables
\begin{equation}\label{eq:Zmimo}
    \forall V_t\subset A(n),\;\;\forall u\in V_t:\quad Z_u(V_t)\dff \left[\sum_{m=1}^{N}(1-\alpha_{u,u}^m)-\sum_{m=1}^{N}(1-\beta_{u,V_t}^m)^+\right]^+\ ,
\end{equation}
and
\begin{equation}\label{eq:Xmimo}
    \forall V_t\subset A(n):\quad X_{V_t}\dff\sum_{u\in V_t}Z_u(V_t)\ .
\end{equation}
Finally, based on the definition of $\dof_{\rm sd}(n,K)$ given in \eqref{eq:dof:sd} and the definition of $X_{V_t}$ we have
\begin{eqnarray}\label{eq:dofmimo:sd1}
    \dof_{\rm sd}(n,K)=\bbe_{\bH}\left[\max_{ V_t\subset A(n):\;| V_t|=K}\; X_{V_t}\right]\ .
\end{eqnarray}
In the following lemmas we find the asymptotic distributions of $Z_u(V_t)$ and $X_{V_t}$, which are instrumental to characterizing  the number of degrees of freedom of interest.
\begin{lemma}\label{lmm:Zmimo}
For the cumulative density function (cdf) of $Z_u(V_t)$, denoted by $F_Z(z)$, for all $V_t\subset A(n)$ and $u\in V_t$ we have
\begin{equation*}
    1-F_Z(z)\doteq  \snr^{-z(z+(K-2)N)}\cdot\mathds{1}_{\{0\leq z\leq N\}}\ .
\end{equation*}
\end{lemma}
\begin{proof}
See Appendix \ref{app:lmm:Zmimo}.
\end{proof}
\begin{lemma}\label{lmm:Xmimo}
For the probability density function (pdf) of $X_{V_t}$, denoted by $F_X(x)$, for all $V_t\subset A(n)$ we have
\begin{equation}\label{eq:Xmimo:cdf}
    1-F_X(x) \doteq \snr^{-x(x+(K-2)N)}\cdot \mathds{1}_{\{0\leq x\leq KN\}}\ .
\end{equation}
\end{lemma}
\begin{proof}
The proof consists of the same line of arguments in the proof of Lemma \ref{lmm:X} and appropriately employing the result of Lemma~\ref{lmm:Xmimo}.
\end{proof}

\subsection{Lower Bounds on $\dof_{\rm sd}(n,K)$}

We find two different lower bounds on the degrees of freedom, and their union provides the desired lower bound.
\begin{description}
  \item[1) $\lim_{n\rightarrow\infty}\dof_{\rm sd}(n,K)\geq\min\left(NK,\zeta_n\right)$:]\textcolor{white}{a}\\
      The proof follows the same line of argument as that of Theorem \ref{th:sd}. Recall the following partitioning of $\{1,\dots,n\}$.
\begin{equation*}
   \forall i\in\{1,\dots,\lfloor n/K\rfloor\}:\quad U_i\dff \{K(i-1)+1,\dots,Ki\}\ .
\end{equation*}
Similar to \eqref{eq:dof:sd:l1} we immediately have
\begin{equation}\label{eq:dof:mimo:sd:l1}
\dof_{\rm sd}(n,K)\geq \bbe_{\bH}\left[\max_{i\in\{1,\dots,n\}}X_{U_i} \right]\ .
\end{equation}
As for $i\neq j$ we have $U_i\cap U_j=\emptyset$ the random variables of $X_{U_i}$ and $X_{U_j}$ are independent. Hence, for the cdf of $\max_iX_{U_i}$, denoted by $F_X^{\max}(x)$, we have
\begin{eqnarray}\label{eq:dof:mimo:sd:l2}
F_X^{\max}(x) = P(\max_i X_{U_i}\leq x) = \big(F_X(x)\big)^{\lfloor n/K\rfloor}= \Big(1-\big(1-F_X(x)\big)\Big)^{\lfloor n/K\rfloor}\ .
\end{eqnarray}
Clearly, as $X_{U_i}\in[0,NK]$ we obtain that
\begin{equation}\label{eq:dof:mimo:sd:l3}
\forall x\geq NK:\quad F^{\rm max}_X(x)=1\ .
\end{equation}
We also use the result of Remark \ref{rem:ab} by setting
\begin{equation*}
f_x(\snr)\dff 1-F_X(x)\quad \overset{\eqref{eq:Xmimo:cdf}}{\Rightarrow}\quad f_x(\snr)\doteq \snr^{-z(z+(K-2)N)}\quad\overset{{\eqref{eq:ab}}}{{\Rightarrow}}\quad a=z(z+(K-2)N)\ ,
\end{equation*}
and
\begin{equation*}
 g_x(\snr)\dff{\lfloor n/K\rfloor}\quad \overset{\eqref{eq:ab}}{\Rightarrow}\quad b=\xi_n\ .
 \end{equation*}
Therefore, from Remark \ref{rem:ab} and \eqref{eq:dof:mimo:sd:l2} we find that for $\max_iX_{U_i}\in[0,NK]$
\begin{equation*}
   \lim_{\snr\rightarrow\infty}F_X^{\rm max}(x)= \lim_{\snr\rightarrow\infty}\left(1-f_x(\snr)\right)^{g_x(\snr)}=\left\{
   \begin{array}{cc}
    0, & \mbox{if}\;\;\xi_n>x(x+(K-2)N, \\
    1, & \mbox{if}\;\;\xi_n<x(x+(K-2)N,
   \end{array}
   \right.\
\end{equation*}
or equivalently for $\max_iX_{U_i}\in[0,NK]$
\begin{equation}\label{eq:dof:mimo:sd:l4}
\lim_{\snr\rightarrow\infty}F_X^{\rm max}(x)= \left\{
\begin{array}{cc}
  0, & \mbox{if}\;\;x<\zeta_n, \\
  1, & \mbox{if}\;\;x>\zeta_n,
\end{array}
\right.\
\end{equation}
where
\begin{equation*}
    \zeta_n\dff\frac{\sqrt{(K-2)^2N^2+4\xi_n}-(K-2)N}{2}\ .
\end{equation*}
By taking into account \eqref{eq:dof:mimo:sd:l3} and \eqref{eq:dof:mimo:sd:l4} and following the same line of argument as in \eqref{eq:dof:sd:l4}-\eqref{eq:dof:sd:l7} we find that
\begin{equation}\label{eq:dof:mimo:sd:l5}
     \lim_{n\rightarrow\infty}\dof_{\rm sd}(n,K)=\lim_{\snr\rightarrow\infty}\dof_{\rm sd}(n,K) \overset{\eqref{eq:dof:Vmimo:sd1}, \eqref{eq:dof:mimo:sd:l1}}{\geq} \lim_{\snr\rightarrow\infty}\bbe_{\bH}\left[\max_iX_{U_i} \right]= \min\left(NK,\zeta_n\right)\ .
\end{equation}
  \item[2) $\lim_{n\rightarrow\infty}\dof_{\rm sd}(n,K)\geq\min\left(NK,\frac{\xi_n}{NK-1}\right)$:]\textcolor{white}{a}\\
      This lower bound can be obtained by directly applying the result of Theorem~\ref{th:sd}. For each user lets us pair up one of the $N$ transmit antennas with one of the $N$ receive antennas and consider them as one pair of transmitter and receiver. Such pairing of the antennas transforms the $(n,K)$-user interference channel with each user equipped with $N$ transmit and receive antennas into an $(Nn,NK)$ interference channels with single-antenna users. According to Theorem~\ref{th:sd} we obtain
      \begin{equation}\label{eq:dof:mimo:sd:l6}
        \min\left(NK,\frac{\xi_n}{NK-1}\right)\leq \lim_{n\rightarrow\infty}\dof_{\rm sd}(n,K)\ .
      \end{equation}
\end{description}
Subsequently, as \eqref{eq:dof:mimo:sd:l5} and \eqref{eq:dof:mimo:sd:l6} provide lower bounds on $\dof_{\rm sd}(n,K)$, their maximum also provides a lower bound $\dof_{\rm sd}(n,K)$. Taking the maximum of these two terms provides the desired lower bound.

\subsection{Proof of Equation \eqref{eq:th:sd:mimo2}}
Recall that according \eqref{eq:dof:mimo:sd:l4}, the pdf of $\max_iU_i$ in the asymptote of large $\snr$ is distributed as a Dirac's delta function. By invoking this fact and following the same line of argument as the proof of \eqref{eq:th:sd2} obtaining the desired result is straightforward.

\section{Conclusions}
\label{sec:conclusion}

The gains of the opportunistic communication in interference channels have been investigated. In particular, we have considered a dense network consisting of $n$ single-antenna transmitter-receiver pairs that affords to activate $K\ll n$ pairs at-a-time. We have shown that by appropriately allocating the resources to $K$ user pairs, when the network size obeys certain scaling laws, it is possible to capture the degrees of freedom within the interval $(\frac{K}{2},K]$ that are not achievable without incorporating opportunistic user activation. We have also generalized the results to the case that the transmitters and receives are equipped with multiple antennas.

\appendix

\section{Proof of Lemma \ref{lmm:Z}}
\label{app:lmm:Z}

By noting that $\{\bH_{u,v}\}$ are distributed as $\N_{\mathbb{C}}(0,1)$, according to Definition \ref{rem:alpha} we have
\begin{equation}\label{eq:alpha:uv:dist}
    \forall u,v\in V_t:\quad f_{\alpha_{u,v}}(\alpha)\doteq \snr^{-\alpha}\cdot\mathds{1}_{\{\alpha\geq 0\}}\ .
\end{equation}
By further defining
\begin{equation}\label{eq:beta:uv}
    \forall u,v\in V_t:\quad \beta_{u,v}\dff \min\{\alpha_{u,v},\;1\}\ ,
\end{equation}
from \eqref{eq:beta:Vt} and \eqref{eq:beta:uv} we have
\begin{equation}\label{eq:beta:Vt2}
     \forall V_t\subset A(n),\;\;\forall u\in V_t:\quad\beta_u({V_t})=\min_{v\in V_t,\;v\neq u}\beta_{u,v}\ .
\end{equation}
Note that due to the statistical independence of the channel coefficients $\{\bH_{u,v}\}$, the elements of $\{\alpha_{u,v}\}_{u,v}$ defined in \eqref{eq:alpha:uv} are also independent. Likewise, the elements of $\{\beta_{u,v}\}_{u,v}$ also become independent. Therefore, from Remark \ref{rem:beta} we find that the cdf of $\beta_u(V_t)$, denoted by $F_{\beta_u}(\beta)$ is given by
\begin{eqnarray}\label{eq:beta:Vt:dist}
    \nonumber F_{\beta_u}(\beta)&=&P\left(\min_{v\in V_t,\;v\neq u}\beta_{u,v}\leq\beta\right)\\
    \nonumber&=&1-\prod_{v\in V_t,\;v\neq u}P\left(\beta_{u,v}\geq\beta\right)\\
    &\overset{\eqref{eq:beta:cdf}}{=}&1-\mathds{1}_{\{\beta< 1\}}\cdot\left(1-\exp\left(-\snr^{-\beta}\right)\right)^{|V_t-1|}\ .
\end{eqnarray}
By recalling that
\begin{equation*}
    Z_u(V_t)=\big(\beta_u(V_t)-\alpha_{u,u}\big)^+\ ,
\end{equation*}
for the cdf of $Z_u(V_t)$, denoted by $F_Z(z)$, we have
\begin{eqnarray}\label{eq:X:cdf1}
  \nonumber F_Z(z) &=& P(Z_u(V_t)\leq z)= 1-P\left((\beta_u({V_t})-\alpha_{u,u})^+> z\right)\\
  \nonumber &=&1-\left[\mathds{1}_{\{z< 0\}}+\mathds{1}_{\{z\geq 0\}}\cdot P\left(\beta_u({V_t})-\alpha_{u,u}> z\right)\right]\\
  \nonumber &=&\mathds{1}_{\{z\geq 0\}}\cdot\left[1-\int_{-\infty}^\infty P\left(\beta_u({V_t})> z+\alpha\right)\;f_{\alpha_{u,u}}(\alpha)\;d\alpha\right]\\
  &\overset{\eqref{eq:beta:Vt:dist}}{=}& \mathds{1}_{\{z\geq 0\}}\cdot\left[1- \int_{-\infty}^\infty\mathds{1}_{\{z+\alpha< 1\}}\cdot\left(1-\exp\left(-\snr^{-(z+\alpha)}\right)\right)^{K-1} f_{\alpha_{u,u}}(\alpha)\;d\alpha\right]\ .
\end{eqnarray}
Therefore, the pdf of $Z_u(V_t)$, denoted by $f_Z(z)$, is given by
\begin{eqnarray}\label{eq:X:cdf2}
  \nonumber f_Z(z) &\overset{\rm (a)}{=}& \delta(z)\cdot \left[1- \int_{-\infty}^\infty\mathds{1}_{\{z+\alpha< 1\}}\cdot\left(1-\exp\left(-\snr^{-(z+\alpha)}\right)\right)^{K-1} f_{\alpha_{u,u}}(\alpha)\;d\alpha\right]
  \\
  \nonumber && +\mathds{1}_{\{z\geq 0\}}\cdot\left[\int_{-\infty}^\infty\delta(z-(1-\alpha))\cdot\left(1-\exp\left(-\snr^{-(z+\alpha)}\right)\right)^{K-1} f_{\alpha_{u,u}}(\alpha)\;d\alpha\right]\\
  \nonumber && +\mathds{1}_{\{z\geq 0\}}\cdot\left[\int_{-\infty}^\infty\mathds{1}_{\{z+\alpha< 1\}}\cdot\snr^{-(z+\alpha)}\left(1-\exp\left(-\snr^{-(z+\alpha)}\right)\right)^{K-2} f_{\alpha_{u,u}}(\alpha)\;d\alpha\right]\\
  \nonumber && \hspace{1 in}\times (K-1)\log\snr\\
  \nonumber &\overset{\rm (b)}{=}&  \delta(z)\cdot \Bigg[1- \int_{-\infty}^1\underset{\doteq\snr^{-(K-1)\alpha}\;{\tiny \rm (Remark~ \ref{rem:a})}} {\underbrace{\left(1-\exp\left(-\snr^{-(\alpha)}\right)\right)^{K-1}}} f_{\alpha_{u,u}}(\alpha)\;d\alpha\Bigg]
  \\
  \nonumber && +\mathds{1}_{\{z\geq 0\}}\cdot\Bigg[\underset{\doteq\snr^{-(K-1)}\;{\tiny \rm (Remark~ ~\ref{rem:a})}}{\underbrace{\left(1-\exp\left(-\snr^{-(1)}\right)\right)^{K-1}}} f_{\alpha_{u,u}}(1-z)\Bigg]\\
  \nonumber && +\mathds{1}_{\{z\geq 0\}}\cdot\Bigg[\int_{-\infty}^{1-z}\underset{\doteq \snr^{-(K-1)(z+\alpha)}\;\;{\tiny \rm (Remark~ \ref{rem:a})}}{\underbrace{(K-1)\log\snr\cdot\snr^{-(z+\alpha)}\left(1-\exp\left(-\snr^{-(z+\alpha)}\right)\right)^{K-2}}} f_{\alpha_{u,u}}(\alpha)\;d\alpha\Bigg]\\
  \nonumber &\overset{\rm (c)}{=}&  \delta(z)-\delta(z)\cdot\int_0^1 \snr^{-(K-1)\alpha}\;\snr^{-\alpha}\;d\alpha\\
  \nonumber && +\mathds{1}_{\{z\geq 0\}}\cdot\snr^{-(K-1)}\cdot\snr^{-(1-z)}\cdot \mathds{1}_{\{z\leq 1\}}\\
  \nonumber && +\mathds{1}_{\{z\geq 0\}}\cdot\Bigg[\int_0^{1-z}\snr^{-(K-1)(z+\alpha)} \snr^{-\alpha}\;d\alpha\Bigg]\\
  \nonumber &\overset{\rm (d)}{\doteq}& \delta(z)- \delta(z)\cdot\snr^0\\
  \nonumber && + \snr^{-(K-1)} \;\snr^{-(1-z)}\;\cdot\mathds{1}_{\{0\leq z\leq 1\}}\\
  \nonumber && + \int_0^{1-z}\snr^{-(\alpha)}\;\snr^{-(z+\alpha)(K-1)} d\alpha\;\cdot\mathds{1}_{\{0\leq z\leq 1\}}\; \\
  \nonumber &\overset{\rm (e)}{=}& \left(\snr^{-(K-z)} + \snr^{-(K-1)z}\right)\cdot\mathds{1}_{\{0\leq z\leq 1\}} \\
  \nonumber  &\overset{\rm (f)}{\doteq} & \snr^{-(K-1)z}\cdot\mathds{1}_{\{0\leq z\leq 1\}}\ .
\end{eqnarray}
Equation (a) is obtained by taking the derivative of \eqref{eq:X:cdf1} with respect to $z$ and (b) is obtained by some simplifications brought by the Dirac's delta function. The exponential equality in Equation (c) holds by recalling Remark \ref{rem:a} and replacing the relevant terms by their exponentially equivalent terms. Equations (d) and (e) hold by finding the dominant integrands that characterize the exponential order of the two integrals. Finally, (f) is obtained by noting that for $0\leq z\leq 1$ we have $(K-1)z\leq K-z$, and thereof the dominant term in (e) is $\snr^{-(K-1)z}$ (Remark~\ref{rem:d}).

\section{Proof of Lemma \ref{lmm:exchange}}
\label{app:lmm:exchange}
Since $X_1,\dots,X_n$ is a sequence of exchangeable random variables, by a result of {\em de Finetti}'s theorem \cite{Heath:76} we know that there exists a random variable $Y$ such that
\begin{equation}\label{eq:ex1}
    P(X_1\leq x_1,\dots,X_n\leq x_n)=\bbe_Y\left[\prod_{i=1}^nP(X_i\leq x_i\med Y)\right]\ ,
\end{equation}
and the conditional random variables $\{X_1\med Y,\dots,X_n\med Y\}$ are {\em identically} distributed. Therefore, we have
\begin{eqnarray*}
  P(\max_i X_i\leq x) &=& P(X_1\leq x,\dots,X_n\leq x)\\
  &\overset{\eqref{eq:ex1}}{=} & \bbe_Y\left[\prod_{i=1}^nP(X_i\leq x\med Y)\right]\\
  &=&\bbe_Y\left[P(X_i\leq x\med Y)\right]^n\\
  &\geq & \Big[\bbe_Y\left[P(X_i\leq x\med Y)\right]\Big]^n\\
  &=&\left[\int_yP(X_i\leq x\med Y)f_Y(y)\;dy\right]^n\\
  &=&\Big[P(X_i\leq x)\Big]^n\ .
\end{eqnarray*}

\section{Proof of Lemma \ref{lmm:beta_alpha}}
\label{app:lmm:beta_alpha}

From the definition of $Z_u(V_t)$ and $W_u$ we clearly have
      \begin{equation}\label{eq:dof:sd1:max4}
        \forall w\geq 1:\quad P\left( \left(\beta_{u,v}-\alpha_{u,u}\right)^+\leq w\right)=1\ .
     \end{equation}
     Furthermore, by following the same line of argument as in \eqref{eq:X:cdf1} we find that
      \begin{align}\label{eq:dof:sd1:max5}
        \nonumber 1- &P\left( \left(\beta_{u,v}-\alpha_{u,u}\right)^+\leq w\right)\\
        \nonumber &= \mathds{1}_{\{w< 0\}}+\mathds{1}_{\{w\geq 0\}}\cdot P\left(\beta_{u,v}-\alpha_{u,u}> w\right)\\
        \nonumber &=\mathds{1}_{\{w< 0\}}+\mathds{1}_{\{w\geq 0\}}\cdot\int_{-\infty}^\infty P\left(\beta_{u,v}> w+\alpha\right)\;f_{\alpha_{u,u}}(\alpha)\;d\alpha\\
        \nonumber &= \mathds{1}_{\{w< 0\}}+\mathds{1}_{\{w\geq 0\}}\cdot\int_{-\infty}^\infty\mathds{1}_{\{w+\alpha< 1\}}\cdot\left(1-\exp\left(-\snr^{-(w+\alpha)}\right)\right) f_{\alpha_{u,u}}(\alpha)\;d\alpha\\
        \nonumber &\overset{\eqref{eq:alpha:pdf:exp}}{\doteq} \mathds{1}_{\{w< 0\}}+\mathds{1}_{\{w\geq 0\}}\cdot\int_{0}^{1-w}\underset{\doteq \snr^{-(w+\alpha)}\;\;({\rm Remark}~\ref{rem:a})}{\underbrace{\left(1-\exp\left(-\snr^{-(w+\alpha)}\right)\right) }} \snr^{-\alpha}\;d\alpha\\
        \nonumber &\doteq \mathds{1}_{\{w< 0\}}+\mathds{1}_{\{w\geq 0\}}\cdot\int_{0}^{1-w}\snr^{-(w+2\alpha)}\;d\alpha\\
        \nonumber &\overset{\eqref{eq:region}}{\doteq } \mathds{1}_{\{w< 0\}}+\mathds{1}_{\{0\leq w\leq 1\}}\cdot\snr^{-b}\quad\mbox{where}\quad b=\inf_{0\leq\alpha\leq 1-w}(w+2\alpha)=w\\
        &\doteq \snr^0\cdot \mathds{1}_{\{w< 0\}}+\snr^{-w}\cdot \mathds{1}_{\{0\leq w \leq 1\}}\ .
      \end{align}
\section{Proof of Lemma \ref{lmm:Zmimo}}
\label{app:lmm:Zmimo}

We start by providing the following two lemmas. These lemmas are closely related to the results existing in \cite{Zheng:IT03} with very slight differences. The proofs, however, are very similar and are omitted for brevity.
\begin{lemma}\cite[Lemma 3]{Zheng:IT03}
\label{lmm:eigen1}
Let $\bA$ be an $p\times q$ random matrix with i.i.d. $\N_{\mathbb{C}}(0,\sigma^2)$ entries. Suppose $p\leq q$ and $\mu_1,\dots,\mu_p$ be the ordered non-zero eigenvalues of $\bA\bA^H$. By defining
\begin{equation*}
    \alpha_m=-\frac{\log\mu_m}{\log\snr}\ ,
\end{equation*}
the joint pdf of the random vector $\balpha=[\alpha_1,\dots,\alpha_p]$ is asymptotically equal to
\begin{equation*}
    f_{\balpha}(\balpha)\doteq\kappa_{p,q}\;(\log\snr)^p \prod_{m=1}^p\snr^{-(q-p+1)\alpha_m}\prod_{m<r}(\snr^{-\alpha_m}-\snr^{-\alpha_r}) \exp\left[-\sum_{m=1}^p\snr^{-\alpha_m}\right]\ ,
\end{equation*}
and for any arbitrary region ${\cal A}$ and $d\geq 0$ we have
\begin{equation*}
    \int_{\cal A}\snr^{-d}\cdot f_{\balpha}(\balpha)\;d\balpha=\int_{\cal A'} \snr^{-d}\cdot\prod_{m=1}^p\snr^{-(2m-1+q-p)\alpha_m} \;d\balpha\ .
\end{equation*}
where
\begin{equation*}
    {\cal A'}=\{\balpha\med \balpha\in {\cal A} \;\;\mbox{and}\;\; \balpha\succeq 0\}\ .
\end{equation*}

\end{lemma}
\begin{lemma}\cite[Theorem 4]{Zheng:IT03}
\label{lmm:eigen2}
Let $\bA$ be an $p\times q$ random matrix with i.i.d. $\N_{\mathbb{C}}(0,\sigma^2)$ entries. Suppose $p\leq q$ and $\mu_1,\dots,\mu_p$ be the ordered non-zero eigenvalues of $\bA\bA^H$. By defining
\begin{equation*}
    \alpha_m=-\frac{\log\mu_m}{\log\snr}\ ,
\end{equation*}
we have
\begin{equation*}
    \forall r\in[0,p]:\qquad P\left[\sum_{m=1}^p(1-\alpha_m)^+<r\right]\doteq \snr^{-d(r)}\ ,
\end{equation*}
where
\begin{equation*}
    d(r)=(q-r)(p-r)\ .
\end{equation*}
\end{lemma}
Now, by taking into account Remark \ref{rem:alpha} and Equation \eqref{eq:alpha:pdf:exp}, the probability that the exponential order terms $\{\alpha_{u,u}^m\}$ and $\{\beta_{u,V_t}^m\}$ are negative is zero. Therefore, by recalling that
\begin{equation*}
    \forall V_t\subset A(n),\;\;\forall u\in V_t:\quad Z_u(V_t)\dff \left[\sum_{m=1}^{N}(1-\alpha_{u,u}^m)-\sum_{m=1}^{N}(1-\beta_{u,V_t}^m)^+\right]^+\ ,
\end{equation*}
we know that with probability 1, random variable $z$ lies in the interval $[0,N]$. Therefore we have
\begin{equation}\label{eq:Z:cdf1}
    1-F_Z(z)=\left\{
    \begin{array}{cc}
      1, & z<0, \\
      0, & z>N.
    \end{array}
    \right.\
\end{equation}
Furthermore, for $0\leq z\leq N$ we have
\begin{eqnarray}\label{eq:Z:cdf1_2}
  \nonumber 1-F_Z(z) &=& P(Z_u(V_t)> z)\\
  \nonumber &=& P\left(\sum_{m=1}^{N}(1-\alpha_{u,u}^m)-\sum_{m=1}^{N}(1-\beta_{u,V_t}^m)^+>z\right)\\
  &=& P\left(\sum_{m=1}^{N}(1-\beta_{u,V_t}^m)^+<N-z-\sum_{m=1}^{N}\alpha_{u,u}^m\right)\ .
\end{eqnarray}
By denoting the pdf of the random vector $\balpha_u=[\alpha_{u,u}^1,\dots,\alpha_{u,u}^N]^T$ by $f_{\balpha_u}(\balpha)$, \eqref{eq:Z:cdf1_2} implies
\begin{eqnarray}\label{eq:Z:cdf3}
   1-F_Z(z) & = & \int_{\balpha} P\left(\sum_{m=1}^{N}(1-\beta_{u,V_t}^m)^+<N-z-\boldsymbol{1}^T\cdot\balpha\right) f_{\balpha_u}(\balpha)\; d\balpha\ .
\end{eqnarray}
Note that $\{\beta_{u,V_t}^m\}_m$ are the eigenvalues of the matrix $\tilde\bH_{u,V_t}\tilde\bH_{u,V_t}^H$. Since all the entries of the $N\times(K-1)N$ matrix $\tilde\bH_{u,V_t}$ are i.i.d. with distribution $\N_{\mathbb{C}}(0,\lambda_{u,V_t})$ according to Lemma~\ref{lmm:eigen2} we have
\begin{equation}\label{eq:Z:cdf2}
    \forall r\in[0,N]:\qquad P\left[\sum_{m=1}^N(1-\beta_{u,V_t}^m)^+<r\right]\doteq \snr^{-(N-r)((K-1)N-r)}\ .
\end{equation}
Equations \eqref{eq:Z:cdf3}-\eqref{eq:Z:cdf2} subsequently give rise to
\begin{eqnarray}\label{eq:Z:cdf3_2}
    1-F_Z(z)
   &\doteq& \int_{\cal A} \snr^{-d(z)} f_{\balpha_u}(\balpha)\; d\balpha
\end{eqnarray}
where
\begin{equation*}
    d(z)=(z+\boldsymbol{1}^T\cdot\balpha)\left((K-2)N+z+\boldsymbol{1}^T\cdot\balpha\right)\ ,
\end{equation*}
and
\begin{equation*}
    {\cal A}=\{\balpha\med \boldsymbol{1}^T\cdot\balpha\leq N-z\}\ .
\end{equation*}
Note that region ${\cal A}$ is characterize by noting that $\sum_{m=1}^N(1-\beta_{u,V_t}^m)^+\geq 0$ and also finding the region of $\balpha$ for which the integrand of \eqref{eq:Z:cdf3_2} is non-zero. Therefore, by using Lemma \ref{lmm:eigen1} from \eqref{eq:Z:cdf3} we further find
\begin{eqnarray*}
   \nonumber 1-F_Z(z)  &\doteq& \int_{\cal A} \snr^{-d(z)} \prod_{m=1}^N\snr^{-(2m-1+(K-2)N)\alpha_m}\; d\balpha\\
   \nonumber & = & \int_{\cal A} \snr^{-d(z)} \snr^{-\sum_{m=1}^N(2m-1+(K-2)N)\alpha_m}\; d\balpha\\
   & \doteq & \snr^{-b}\ ,
\end{eqnarray*}
where
\begin{eqnarray*}
    b & = & \inf_{\balpha\in{\cal A'}}\left[d(z)+\sum_{m=1}^N(2m-1+(K-2)N)\alpha_m \right]\\
    & = & \inf_{\balpha\in{\cal A'}}\left[(z+\boldsymbol{1}^T\cdot\balpha)\left((K-2)N+z+\boldsymbol{1}^T\cdot\balpha\right) +\sum_{m=1}^N(2m-1+(K-2)N)\alpha_m \right]\ .
\end{eqnarray*}
Clearly the minimum of the term above occurs when $\balpha=\boldsymbol{0}$. Therefore, for $0\leq Z\leq N$
\begin{equation*}
    1-F_Z(z)\doteq \snr^{-b}\quad\mbox{where}\quad b=z(z+(K-2)N)\ ,
\end{equation*}
which in conjunction with \eqref{eq:Z:cdf1} concludes the proof.


{\small\bibliographystyle{IEEEtran}
\bibliography{IEEEabrv,Opportunistic_GIC}}

\end{document}